%% file: fullmain.tex
\newif\iflong\longtrue
\newif\ifproofs\proofstrue

\iflong
\newcommand{\shorten}[2]{#1}
\else
\newcommand{\shorten}[2]{#2}
\fi
\newcommand\figurespace[1]{\vspace*{#1}}

\documentclass[times]{article}
\usepackage{amssymb,times}
\usepackage{xspace}
\usepackage{comment}
\usepackage[cmex10]{amsmath}
\usepackage{graphicx}
\usepackage{color}
\usepackage{paralist}
\usepackage{wrapfig}

\graphicspath{{./figures/}}

\newcommand{\ds}{\displaystyle}

\newcommand{\ap}{\mathbb{AP}}
\newcommand{\ps}[1]{2^{#1}}
\newcommand{\mc}[1]{\mathsf {MC}_{#1}}
\newcommand{\pzero}{{Player~0}\xspace}
\newcommand{\pone}{{Player~1}\xspace}

\newcommand{\zval}{\mathsf{val}}
\renewcommand{\oval}{\mathsf{v}}
\newcommand{\dist}[1]{{\mathcal D}({#1})}
\newcommand{\preoval}{\tilde{\mathsf{v}}}
\newcommand{\F}{\mathsf{ff}}
\newcommand{\T}{\mathsf{tt}}
\newcommand{\ldbr}{[\![}
\newcommand{\rdbr}{]\!]}
\newcommand{\dbr}[1]{\ldbr{#1}\rdbr}
\newcommand{\bstate}[2]{\dbr{#1}_{\bowtie {#2}}}
\newcommand{\br}[1]{[{#1}]}

\newcommand{\cl}[1]{\mathsf{cl}({#1})}

\newcommand{\pair}[1]{{\langle {#1} \rangle}}
\newcommand{\atom}[1]{\mathsf{#1}}
\newcommand{\U}{\,\mathsf{U}\,}
\newcommand{\eq}[1]{{(\!( {#1} )\!)}}
\newcommand{\lang}[1]{{\mathcal L}({#1})}
\newcommand{\val}[2]{{\mathsf{val}({#1},{#2})}}

\newcommand{\prob}[2]{\mathsf{Prob}_{{#1}}({#2})}
\newcommand{\measure}[4]{\mathsf{Msr}_{{#1}}^{{#2}}({#4},{#3})}
\renewcommand{\succ}{\mathsf{succ}}
\newcommand{\dual}[1]{\mathtt{dual}({#1})}
\newcommand{\turn}[1]{\mathtt{turn}({#1})}
\newcommand{\lorc}[2]{{({#1},{#2}]}}

\newcommand{\init}{{\textnormal{in}}}
\newcommand{\set}[1]{\{{#1}\}}
\newcommand{\pog}{{POG}\xspace}
\newcommand{\stam}[1]{}
\newcommand{\cC}{{\cal C}}
\newcommand{\cG}{{\cal G}}
\newcommand{\cN}{{\cal N}}
\newcommand{\cO}{{\cal O}}

\newtheorem{definition}{Definition}
\newtheorem{theorem}{Theorem}
\newtheorem{corollary}{Corollary}
\newtheorem{example}{Example}

\newtheorem{lemma}{Lemma}

\def\squarebox#1{\hbox to #1{\hfill\vbox to #1{\vfill}}}

\newcommand{\qed}{\hspace*{\fill}
	    \vbox{\hrule\hbox{\vrule\squarebox{.667em}\vrule}\hrule}\smallskip}
\newenvironment{proof}{\begin{trivlist}
\item[\hspace{\labelsep}{\bf\noindent Proof: }]
}{\qed\end{trivlist}}

\newcommand{\mysection}{\section}
\newcommand{\mysubsection}{\subsection}

\begin{document}
\title{Obligation Blackwell Games and p-Automata}

\author{Krishnendu Chatterjee\\
Institute of Science and Technology Austria
\and
Nir Piterman\\
University of Leicester
}



\maketitle

\begin{abstract}
We recently introduced p-automata, automata that read discrete-time
Markov chains and showed they provide an automata-theoretic framework
for reasoning about pCTL model checking and abstraction of discrete
time Markov chains.
We used turn-based stochastic parity games to
define acceptance of Markov chains by a special subclass of p-automata.
Definition of acceptance required a reduction to a
series of turn-based stochastic parity games.
The reduction was cumbersome and complicated and could not
support acceptance by general p-automata, which was left undefined as
there was no notion of games that supported it.

Here we generalize two-player games by adding a structural
acceptance condition called \emph{obligations}.
Obligations are orthogonal to the linear winning conditions that
define whether a play is winning.
Obligations are a declaration that player 0 can achieve a certain
value from a configuration.
If the obligation is met, the value of that configuration for player 0
is~1.

One cannot define value in obligation games by the standard mechanism
of considering the measure of winning paths on a Markov chain and
taking the supremum of the infimum of all strategies.
Mainly because obligations need definition even for Markov chains and
the nature of obligations has the flavor of an infinite nesting of supremum 
and infimum operators.
We define value via a reduction to turn-based games
similar to Martin's proof of determinacy of Blackwell games with Borel
objectives.
Based on this value definition we show that obligation games are
determined.
We show that for Markov chains with Borel objectives and obligations,
and finite turn-based stochastic parity games with obligations there
exists an alternative and simpler characterization of the value
function without going through a Martin-like reduction. 
Based on this simpler definition we 
give
an exponential time algorithm to analyze finite turn-based stochastic
parity games with obligations and show that the decision problem of 
winning parity games with obligations is in NP$\cap$co-NP.
Finally, we show that obligation games provide the necessary framework
for reasoning about p-automata and that they generalize the
previous definition.
\end{abstract}

\mysection{Introduction}
\label{section:introduction}
Markov chains are a very important modeling formalism in many areas of
science. 
In computer science, Markov chains form the basis of central techniques such
as performance modeling, and the design and correctness of randomized
algorithms used in security and communication protocols.
Recognizing this prominent role of Markov chains, the formal-methods
community has devoted significant attention to these models, e.g., in
developing model checking for {\em qualitative}
\cite{HSP82,CY95,Var85b} and {\em quantitative} \cite{ASBBS95}
properties, logics for reasoning about Markov chains \cite{HS86,LJ91},
and probabilistic simulation and bisimulation \cite{LS91,LJ91}. 
Model-checking tools such as PRISM \cite{HKNP06} and LiQuor \cite{CB06}
support such reasoning about Markov chains and have users in many
fields of computer science and beyond.

The automata-theoretic approach to verification has proven to be very
powerful for reasoning about systems modeled as Kripke structures.
For example,
it supports algorithms for satisfiability of temporal logics~\cite{EL85b}, 
model checking~\cite{KVW00}, and abstraction~\cite{HKR02}.

We recently introduced p-automata, which are devices that read
Markov chains as input \cite{HPW10}.
We showed that p-automata provide an automata-theoretic framework
for reasoning about pCTL model checking, and abstraction of discrete
time Markov chains.
The definition of p-automata is motivated by pCTL \cite{HJ94}, the de-facto
standard for model checking Markov chains, and alternating tree
automata:
they combine the rich combinatorial structure of alternating
automata with pCTL's ability to quantify the probabilities
of regular sets of paths.
Acceptance of Kripke
structures by alternating tree automata is decided by solving
turn-based games (cf.~\cite{GTW02}).
Similarly, acceptance of Markov chains 
by p-automata is decided by solving turn-based \emph{stochastic}
games. 
However, acceptance of p-automata was defined through a complicated
and cumbersome reduction to solving a \emph{series} of turn-based
stochastic parity games.
Furthermore, this reduction supported only a subclass of p-automata,
which we called \emph{uniform}, and could not be generalized to
unrestricted p-automata.
Intuitively, uniform automata separate measuring probability
of regular path sets and setting thresholds on these probabilities.
For uniform p-automata, we showed how acceptance can  be decided by a
series of turn-based stochastic parity games.
Acceptance for general p-automata could not be defined as there was no
game framework that supported the unbounded interaction between
measuring probability and setting thresholds.

Here, we propose a game notion that supports such interaction.
In order to do that we augment winning conditions in
games by adding a structural acceptance condition called
\emph{obligations}. 
A winning condition is a combination of a classical set of winning
paths and obligations on some of the game configurations.
An obligation is a declaration by \pzero that she can win with a
certain value from the configuration.
Then, in order to be able to derive a non-zero value from a
configuration with obligation, \pzero has to ensure that the measure
of paths that satisfy the winning condition from that configuration
meets her promised value.
If she can do that, the configuration will have the value 1 for
her. 
That is, the value of the configuration does not depend on the measure
of winning paths obtained in the game following the visit to that
configuration, only on whether the obligation is met or not.
In other words, in order to meet an obligation the measure of the union 
of the following sets of paths must satisfy the value constraint of the 
obligation: paths that (i)~reach other obligations that can be met, or 
(ii)~paths that never reach other obligations and satisfy the winning 
condition. 
In addition, paths that visit infinitely many obligations must satisfy the
winning condition as well.

\begin{figure}[bt]
\begin{center}
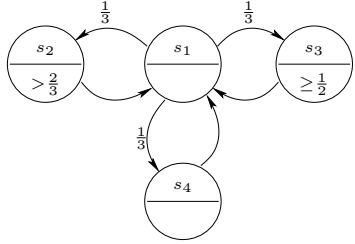 
\end{center}
\figurespace{-3mm}
\caption{\label{figure:intuition}
  Intuition regarding obligations. 
  Circles denote states of Markov chains. 
  The name of the state is written in the top half and the obligation
  (if exists) is written in the bottom half.
  Probabilities are written next to edges except in case there is a
  single outgoing edge whose probability is $1$.
}
\figurespace{-3mm}
\end{figure}
Consider the example in 
Figure~\ref{figure:intuition}. 
Suppose that every path that visits $s_2$ infinitely often is
rejecting and every accepting path that visits $s_4$ infinitely often
must visit $s_3$ infinitely often.
In addition, $s_2$ has an obligation of more than
$\frac{2}{3}$.
That is, in order for $s_2$ to have a positive value,  the measure of the 
set of paths starting in $s_2$ that satisfies the
acceptance condition must be more than $\frac{2}{3}$.
If the obligation is met, the value of $s_2$ is $1$. Otherwise, the
value of $s_2$ is $0$.
Similarly, $s_3$ has an obligation of at least $\frac{1}{2}$.
The obligation at $s_2$ cannot be met.
Every set of paths starting at $s_2$ with measure more than $\frac{2}{3}$ 
must contain a path that reaches $s_2$ again.
Thus, in order to fulfill the obligation of $s_2$, it has to be
visited again, recursively, where the recursion is unfounded.
As the path that visits $s_2$ infinitely often is rejecting, $s_2$'s
obligation cannot be met.
Thus, $s_2$'s value is $0$.
On the other hand, the obligation of $s_3$ can be met. 
Indeed, the measure of paths that start in $s_3$ and reach $s_3$ again
without passing through $s_2$ (i.e., $\neg s_2$ until $s_3$) is 
exactly $\frac{1}{2}$.
Whenever a path reaches $s_3$ the same obligation needs to be met again.
So the same set of paths is used again.
Let us consider the set of infinite paths obtained by concatenating 
infinitely many finite segments (from obligations to obligations,  
i.e., concatenating paths from $s_3$ to itself without visiting  
$s_2$) 
considered above.
These paths visit $s_3$ infinitely often and  never visit $s_2$, 
and thus satisfy the winning condition.

%


\begin{figure}[bt]
\begin{center}
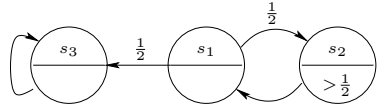 
\end{center}
\figurespace{-3mm}
\caption{\label{figure:measure zero}A set of measure 0 matters.}
\figurespace{-5mm}
\end{figure}
When obligations are involved, the value depends on sets of measure
$0$.
Consider, for example, the Markov chain in 
Figure~\ref{figure:measure zero}. 
Suppose that the set of winning paths includes all paths.
Then, the obligation of configuration $s_2$ can be met.
Indeed, the probability to reach another obligation that can be met
($s_2$) or never reach obligations and win is $1$.
It follows that the value of configuration $s_1$ is $1$.
However, by removing from the set of winning paths the single path
$(s_1\cdot s_2)^\omega$, whose measure is 0, this changes.
The obligation at configuration $s_2$ can no longer be met. 
Every set of paths starting from $s_2$ of measure more than half 
must contain a path that reaches $s_2$ again.
Thus, in order to meet the obligation of $s_2$, the path $(s_1\cdot
s_2)^\omega$ must be included.
This is a losing path that makes infinitely many obligations.
Hence, the value of $s_2$ is $0$ and the value of
$s_1$ is $\frac{1}{2}$.

These two examples consider the most simple types of interaction,
where no player makes choices and everything is determined by chance.
Blackwell games with Borel objectives are a very general form of graph games.
Blackwell games are two-player games where both players choose
simultaneously and independently their actions (also known as
concurrent games).
From every configuration a choice of actions determines a distribution
over successor configurations.
Borel objectives are Borel sets of infinite paths of configurations.
We now discuss the main challenges involved in giving a formal
definition of the value in general obligation Blackwell games.
The typical way of defining a value in a two-player game is done by a
sequence of steps:
First, fixing strategies for the two players we get a Markov chain for
which the measure of winning paths is well defined.
Then, fixing a strategy for \pzero one can get a value of this
strategy by considering the infimum over all strategies for \pone.
Finally, the value of \pzero in the game is the supremum of the
values of all her strategies.
All three steps fail for obligations.
First, it is not even clear how to define the value of a Markov chain, 
and the value depends on sets of measure 0 (see
Figure~\ref{figure:measure zero}).
Then, it is not clear that a strategy for \pzero has a well defined
value.
Such a well defined value would depend on independence
between the choice of ``met obligations'' and the choice of
strategy for \pone. 
Finally, the nature of obligations means that taking the supermum over
strategies of \pzero does not lead to the real value of the game.
This follows because in some types of Borel games (without
obligations) the game value cannot be attained.
There is no single strategy that can achieve the value of the game
but an infinite sequence of strategies can get arbitrarily close to
the value.
A single obligation game requires a player to commit to achieving
values possibly infinitely many times.
Each obligation could be the result of an infinite sequence of
strategies.
Hence, a single external supremum does not capture this and we need an
infinite nesting of supremum (and infimum) operators.

In order to define the value of Blackwell games with
Borel objectives and obligations we use a reduction to turn-based
games with Borel objectives similar to Martin's 
proof that Blackwell games with Borel objectives are determined
\cite{Mar98}.
Intuitively, we explicitly add the value to the game and make \pzero
prove that a value can be won by showing how it propagates under
probabilistic choices made by both players.
We add a new crucial component that captures
obligations to Martin's proof.
This ingredient, which we call \emph{concession}, captures the notion
that the value for \pzero may not be achieveable but may be
approximated arbitrarily close.
To capture approximation \pone chooses a small concession to grant
\pzero after which \pzero should not have a problem to show that she
can win the required amount.
This reduction defines a value for Blackwell Borel obligation games.

A major issue is then that of well definedness of obligation games.
Well definedness, or determinacy, means that whatever \pzero cannot
avoid losing \pone ensures to win and vice versa. 
Formally, it says that the sum of values in a game for \pzero and
\pone is always $1$.
This is a fundamental property that needs to be established for
games.
For almost all types of games used in verification determinacy has
never been an issue.
This relies on Martin's foundational result that Blackwell games with
Borel objectives are determined \cite{Mar98}.
Blackwell games are general enough so that a simple reduction
to them suffices to show determinacy for almost all types of
two-player games.
However, obligations take our games out of scope of Martin's result.
This is apparent as inclusion/removal of a $0$-measure set can change
the value of a game.
In order to show determinacy of Blackwell games with 
Borel objectives and obligations we analyze our reduction to
turn-based games with Borel objectives.
We show that indeed our games are determined. 

Our reduction gives the general definition of values and its
analysis gives us the determinacy result.
However, the reduction is not amenable to computational analysis 
as it 
constructs an uncountable game. 
For computational analysis, we consider 
Markov chains with Borel objectives and obligations and finite turn-based 
stochastic parity games with obligations.
We show that in these cases, we can embed the notion of winning into
the structure of the game by using \emph{choice sets}.
Intuitively, these are the obligations that have value $1$, i.e.,
where the obligation of \pzero is actually met.
This gives rise to a simpler definition where we contrast a strategy
of one player with the strategy of the other player as customary in
definition of games.
We show that this simpler definition, which does not work for the general case, 
coincides with the definition arising from the Martin-like reduction 
for Markov chains and finite turn-based stochastic parity games with obligations. 
Based on this direct characterization, we give algorithms
that analyze 
finite turn-based stochastic parity games with obligations.
We show how to decide whether the value in such a game is at least (or
more) than a given value $r\in [0,1]$ in
$\mbox{NP}{\cap}\mbox{co{-}NP}$ and to compute the value in 
exponential time.
The algorithm identifies a general choice set and calls a solver for
finite turn-based stochastic parity games to check the sanity of
the choice set.
Our $\mbox{NP}{\cap}\mbox{co{-}NP}$ bound matches the bounds for the 
special cases of turn-based stochastic reachability games (without
obligations).

We also show that if games with obligations have a finite number of
exchanges between obligations and no-obligations, then the analysis of
the game can be reduced to the analysis of a series of Blackwell games
(with no obligations).

Finally, we return to p-automata and using turn-based stochastic
parity games with obligations we define
acceptance of general p-automata.
We show that the new definition using the obligation games generalizes
acceptance of uniform p-automata as defined in \cite{HPW10}.

\noindent{\em Related works.}
While we considered an automata theoretic approach to capture pCTL an alternative
approach is to consider probabilistic $\mu$-calculus. 
The problem of considering a probabilistic $\mu$-calculus framework 
to capture pCTL was considered in~\cite{Mio12a,Mio12b}. 
They add an independent choice to turn-based 
stochastic games and show their determinacy and that such games
give a semantics to a probabilistic $\mu$-calculus.
In \cite{MS13}, they show decidability of a fragment of their
probabilistic $\mu$-calculus in 3EXPTIME. 
In contrast, our determinacy result for obligation games is for the 
general class of Blackwell (concurrent) games with Borel objectives,
and our decidability result for turn-based stochastic obligation 
parity games establishes a  $\mbox{NP}{\cap}\mbox{co{-}NP}$ bound.
Comparison of the two types of games does not seem simple.
Our framework for turn-based stochastic obligation parity games 
could also provide better algorithmic analysis for fragments of 
probabilistic $\mu$-calculus of~\cite{Mio12a,Mio12b}.

\mysection{Background}
\label{section:prelim}


For a countable set $S$ let $\dist{S}=\{d:S \rightarrow [0,1]~|~
\exists T\subseteq S \mbox{ such that } |T|\in \mathbb{N}, \forall
s \notin T ~.~ d(s)=0 \mbox{ and }\Sigma_{s\in T} d(s)=1\}$ be the set of
discrete probability distributions with finite support over $S$. 
A distribution $d$ is {\em pure} if there is some $s\in S$ such that
$d(s)=1$.

A \emph{countable labeled Markov chain} $M$ over set of atomic
propositions $\ap$ is a tuple $(S,P,L,s^\init)$, where $S$ is a
countable set of \emph{locations}, $P\colon S \rightarrow {\cal
  D}(S)$ is a probabilistic transition, $s^\init \in S$ the 
\emph{initial} location, 
and $L\colon S\rightarrow \ps \ap$ a \emph{labeling function} with $L(s)$
the set of propositions true in location $s$. 
We sometimes also treat $P$ as a function $P:S\times S \rightarrow
[0,1]$, where $P(s,s')$ is $P(s)(s')$. 
Let $\succ(s)$ be the set $\set{s'\in S \mid P(s,s')>0}$ of 
\emph{successors} of $s$.
By definition all Markov chains we consider are \emph{finitely branching},
i.e.\ $\succ(s)$ is finite for all $s\in S$.
We write $\mc\ap$ for the set of all (finitely branching) Markov
chains over $\ap$. 
A \emph{path} $\pi$ from location $s$ in $M$ is an infinite sequence
of locations $s_0s_1\dots$ with $s_0 = s$ and $P(s_i,s_{i+1}) > 0$ for
all $i\geq 0$. 


Given a Markov chain $M$ with set of states $S$, an open set in
$S^\omega$ is a set $\set{w}\cdot S^\omega$ for some $w\in S^*$.
A set is \emph{Borel} if it is in the $\sigma$-algebra defined by
these open sets.
The measure of every Borel set $\alpha$ is defined as usual in this
$\sigma$-algebra \cite{Bil08,RP10}.
We denote the measure of a Borel set $\alpha$ as $\prob M\alpha$.

\excludecomment{pctldefinition}
\begin{pctldefinition}
Without loss of generality \cite{FHPW09}, one may define the probabilistic 
temporal logic pCTL \cite{HJ94} in 
``Greater Than Negation Normal Form'':
only propositions can be negated and probabilistic bounds
are either $\geq$ or $>$~--~see Fig.~\ref{fig:pctl}. 
\begin{figure}[bt]
\begin{framed}
\figurespace{-6mm}
{\small
\[
\begin{BNFarray}
   \BNFtop{\phi,\psi}        {pCTL formulas}
   \BNFrow{\atom{a},\neg\atom{a}}         {Atom}
   \BNFrow{\phi\land \psi}   {Conjunction}
   \BNFrow{\phi\lor \psi}   {Disjunction}
   \BNFrow{\br{\alpha}_{\bowtie p}}   {Path Probability}
\end{BNFarray}
\begin{BNFarray}
   \BNFtop{\alpha}      {{}\!\!\!\!\! Path formulas}
   \BNFrow{\X \phi}      {\!\!\! Next}
   \BNFrow{\phi \U \psi} {\!\!\! Until}
   \BNFrow{\phi \W \psi} {\!\!\! Weak Until}
\end{BNFarray}
\]
} 
\figurespace{-6mm}
\end{framed}
\figurespace{-5mm}
\caption{Syntax of pCTL, where $\atom{a}\in\ap$,
  $p\in [0,1]$, and $\bowtie{}\in \set{>,\geq}$ \label{fig:pctl}}
\figurespace{-9mm}
\end{figure}

Our semantics of pCTL is as in \cite{HJ94}: path formulas $\alpha$ are
interpreted  
as predicates over paths in $M$, and wrap pCTL formulas into ``LTL''
operators for Next, 
(strong) Until, and Weak Until.
The semantics $\denote {}\phi\subseteq S$ of pCTL formula $\phi$ lifts
path formulas to state formulas: $s\in \denote {}{[\alpha]_{\bowtie
    p}}$ iff 
$\prob Ms\alpha$, the probability of the measurable
set \cite{kemeny76} $\paths s{}{\alpha}$ of paths $ss_1s_2\dots$ in $M$
with $ss_1s_2\dots \models \alpha$, satisfies $\bowtie p$.
$M$ satisfies  $\phi$, denoted $M\models \phi$, if $s^\init \in
\denote {}\phi$.
\end{pctldefinition}

\paragraph{Blackwell Games}

A {\em Blackwell game} is $G=(V,A_0,A_1,R,\alpha)$, where $V$ is a
countable set of configurations, $A_0$ and $A_1$ are finite sets of
actions, $\alpha$ is a Borel set defining the winning set of \pzero,
and $R:V \times A_0 \times A_1 \rightarrow \dist{V}$ is a
transition function associating with a configuration $v$ and a pair of
actions for both players a distribution over next configurations with
\emph{finite support}.
A play is an infinite sequence $p=v_0v_1\cdots$ such that for every
$i\geq 0$ there are $a^0_i\in A_0$ and $a^1_i\in A_1$ such that
$R(v_i,a_i^0,a^1_i)(v_{i+1})>0$. 

A {\em strategy} for \pzero is 
$\sigma:V^+\rightarrow \dist{A_0}$. 
A strategy for \pone is similar.
A strategy is \emph{memoryless} if for every $w,w'\in V^*$ and $v\in
V$ we have $\sigma(wv)=\sigma(w'v)$ and
it is \emph{pure} if for every $w\in V^+$ we have $\sigma(w)$
is pure.
Let $\Sigma$ (resp.\ $\Pi$) be the set of all strategies for \pzero
(resp.\ \pone).

Each $(\sigma,\pi) \in \Sigma\times\Pi$ from game $G$ and
configuration $v$ determine a Markov chain with locations $V^+$.
Formally, $v(\sigma,\pi)=(V^+,P,L,v)$, where the labeling function $L$ is 
irrelevant, and for every $w\in V^*$ and $v'\in V$ we set
$P(vwv')=\sum_{a_0 \in A_0} 
\sum_{a_1\in A_1}
\sigma(vwv')(a_{0}) \cdot \pi(vwv')(a_1) \cdot R(v',a_0,a_1)
$.

Sometimes, we may want to start a game from an initial sequence of
configurations, which we call \emph{play prefix} or just
\emph{prefix}.
Let $w=v_0\cdots v_n \in V^+$ be a prefix.
Then $w(\sigma,\pi)$ is the Markov chain $(\set{w}\cdot V^*,P,L,w)$,
where 
$P(wuv)=\sum_{a_0\in A_o}\sum_{a_1\in A_1}
\sigma(wuv)(a_0)\cdot \pi (wuv)(a_1) \cdot R(v,a_0,a_1)$,
for $u\in V^*$ and $v\in V$.
All definitions, generalize to this setting.

The value of $(\sigma,\pi)$ for \pzero from prefix $w\in
\set{v}\cdot V^*$, is
$\prob{w(\sigma,\pi)}{(\set{w}\cdot V^\omega) \cap
    \alpha}$,
denoted $\zval_0(v(\sigma,\pi),w)$.
The value of $w$ for \pzero in $G$ is 
$\ds \sup_{\sigma\in \Sigma}
\ds \inf_{\pi\in \Pi}
\zval_0(w(\sigma,\pi),w)$, denoted $\zval_0(G,w)$.
Dually, the value of $w$ for \pone in $G$, denoted 
$\zval_1(G,w)$, is $\ds \sup_{\pi\in \Pi}
\ds \inf_{\sigma\in \Sigma}
(1-\zval_0(w(\sigma,\pi),w))$.

\begin{theorem}
Let $G$ be a game and $\alpha$ a Borel set.
Then for every $w\in V^+$ we have $\zval_0(G,w)+\zval_1(G,w)=1$ 
{\rm \cite{Mar98}}.
\label{theorem:determinacy blackwell games}
\end{theorem}

The value $\zval_1(G,w)$ can be also
obtained by considering the game
$\dual{G}=(V,A_1,A_0,\dual{R},V^\omega{\setminus}\alpha)$, where 
$\dual{R}(v',a_1,a_0)=R(v',a_0,a_1)$.
Formally, $\zval_1(G,w)=\zval_0(\dual{G},w)$. 
By definition, for every Markov chain $M$ and every measurable set $
\alpha\subseteq V^\omega$ we have $\prob M\alpha = 1-\prob
M{V^\omega{\setminus}\alpha}$. 

\paragraph{Turn-Based Stochastic Games}

A turn-based stochastic game $G$ is
$((V,E),(V_0,V_1,V_p),\kappa,\alpha)$ with the following 
components. 
\begin{compactitem}
\item
$V$ is a countable set of configurations.
\item
$E \subseteq V^2$ is a set of edges such that for every $v\in V$ we
have $|\{v' ~|~ (v,v')\in E\}|$ is finite.
\item
The triplet $(V_0,V_1,V_p)$ partitions $V$ so that $V_0$ is the set of
\pzero configurations, $V_1$ is the set of \pone configurations, and
$V_p$ is the set of probabilistic configurations.
\item
$\kappa: V_p\rightarrow \dist{V}$ is such that $\kappa(v)(v')>0$ if
and only if $(v,v')\in E$.
\item
$\alpha$ is a Borel set as before.
\end{compactitem}
A play is an infinite sequence $v_0v_1\cdots$ such that for all $i\in
\mathbb{N}$ we have $(v_i,v_{i+1})\in E$.
A strategy for 
\pzero is a function $\sigma \colon V^*\cdot V_0 \rightarrow
\dist{V}$ such that for all $w\in V^*$ and $v\in V_0$ we have
$\sigma(wv)(v')>0$ implies $(v,v')\in E$.
Strategies for \pone are defined analogously. 
The type of strategy is determined by the type of game and no
confusion will arise.
As before $(\sigma,\pi)\in \Sigma\times\Pi$ determine a Markov chain
$w(\sigma,\pi)$. 
Then, the value of \pzero from prefix $w$ is 
$
\zval_0(G,w)=
\sup_{\sigma\in\Sigma}\inf_{\pi\in\Pi}
\prob{w(\sigma,\pi)}{\alpha}$
and the value of \pone from prefix $w$ is
$
\zval_1(G,w)=
\sup_{\pi\in\Pi}\inf_{\sigma\in\Sigma}(1-\prob{w(\sigma,\pi)}{\alpha})$.

A turn-based stochastic game can be seen as a Blackwell game whose
configurations are of the following types:
\begin{compactitem}
\item
$v$ is a \pzero configuration if for every $a_0\in A_0$ and
$a_1,a'_1\in A_1$ we have $R(v,a_0,a_1)=R(v,a_0,a'_1)$ and
$R(v,a_0,a_1)$ is pure.
\item
$v$ is a \pone configuration if for every $a_0,a'_0\in A_0$ and
$a_1\in A_1$ we have $R(v,a_0,a_1)=R(v,a'_0,a_1)$ and $R(v,a_0,a_1)$
is pure.
\item
$v$ is a probabilistic configuration if for every $a_0,a'_0\in A_0$
and $a_1,a'_1\in A_1$ we have $R(v,a_0,a_1)=R(v,a'_0,a'_1)$.
\end{compactitem}

\begin{corollary}
Let $G$ be a turn-based stochastic game, $\alpha$ a Borel set, and
$w\in V^+$.
Then, $\zval_0(G,w)=1-\zval_1(G,w)$.
\label{corollary:determinacy stochastic games}
\end{corollary}

When $V_p=\emptyset$ the game is \emph{simple} or just 
\emph{turn based}.
For turn-based games the set $V$ does not have to be countable.
In this case it is enough to consider pure strategies, which implies
that a pair $(\sigma,\pi)\in \Sigma\times\Pi$ induces a unique play
$w(\sigma,\pi)$.
Then, the value of \pzero from prefix $w$ is either $1$ or $0$.
Equivalently, \pzero wins from $w$ if there is a strategy $\sigma$
such that for every strategy $\pi$ we have $w(\sigma,\pi)\in \alpha$.
Otherwise, \pone wins from $w$.
In this case, we write $W_0=\set{w ~|~ \mbox{\pzero wins from $w$}}$ and
$W_1= \set{w ~|~ \mbox{\pone wins from $w$}}$. 

\begin{theorem}
Let $G$ be a turn-based game and $\alpha$ a Borel set.
Then $W_0\cap W_1=\emptyset$ and $W_0 \cup W_1=V^+$.
{\rm \cite{Mar75}}
\label{theorem:determinacy turn based games}
\end{theorem}

When $w\in W_0$ we also write $\zval_0(G,w)=1$ and $\zval_1(G,w)=0$.
Dually, when $w\in W_1$ we write $\zval_0(G,w)=0$ and $\zval_1(G,w)=1$.

We say that $\alpha$ is derived from a parity condition
$c:V\rightarrow [0..k]$ if for every play $p=v_0v_1v_2\cdots \in
V^\omega$ we have $p\in \alpha$ iff $\ds\liminf_{n\to \infty} c(v_n)$ is even.

\begin{theorem}
Consider a finite game $G$, where $\alpha$ is derived from a parity
condition, ${\bowtie}\in\{>,\geq\}$, and $r$ is a rational.
\begin{compactitem}
\item
If $G$ is a Blackwell game, whether $\zval_0(G,w)\bowtie r$ and
$\zval_1(G,w) \bowtie r$ can be decided in PSPACE {\rm \cite{Cha07}}.
\item
If $G$ is a turn-based stochastic game, the values $\zval_0(G,w)$ and
$\zval_1(G,w)$ can be computed in exponential time and whether
$\zval_i(G,w) \bowtie r$ can be decided in
$\mbox{NP}{\cap}\mbox{co{-}NP}$ {\rm \cite{CJH04}}.
\item
If $G$ is a turn-based stochastic game, there is a
memoryless strategy 
$\sigma$ achieving $G$'s value {\rm \cite{CJH04}}.
That is:
\begin{center}
$\inf_{\pi\in\Pi}\prob{w(\sigma,\pi)}{\alpha}=\zval_0(G,w)$
\end{center}
\item
If $G$ is a turn-based game, whether $w\in W_0$ can be decided in 
UP$\cap$co-UP {\rm \cite{Jur98}}.
\item
If $G$ is a countable turn-based game
then 
there are pure-memoryless strategies $\sigma$ and $\pi$ such that
$\sigma$ is winning from every configuration $v\in W_0$ and $\pi$ is
winning from every configuration $v\in W_1$ {\rm \cite{Tho97}}.
\end{compactitem}
\label{theorem:complexity parity games}
\label{theorem:memoryless determinacy of parity games}
\end{theorem}

%

\mysection{Obligation Blackwell Games}
\label{section:obligation blackwell games}
We introduce obligation Blackwell games.
These games extend Blackwell games by having a winning condition that
includes a winning set (as in normal Blackwell games) and a
set of obligations.
Intuitively, a play is winning for \pzero if it belongs to the winning
set.
However, whenever meeting an obligation, \pzero has to make sure that
the value of the game in that configuration satisfies the obligation.
If the obligation can be met, the value for \pzero at the configuration
is~1. 

An obligation Blackwell game (OBG for short) is $G=(V,A_0,A_1,R,\cG)$,
where $V$, $A_0$, $A_1$, and $R$ are like in Blackwell games.
The goal $\cG=\pair{\alpha,O}$, where $\alpha \subseteq V^\omega$
is a Borel set as for Blackwell games and 
$O:V \rightarrow (\set{\geq, >} 
\times [0,1]) \cup \set{\bot}$.
The obligation function $O$ associates with some configurations the
value $\bot$ saying that there is no obligation associated
with this configuration. 
With other configurations $O$ associates an obligation ${>}r$ or
${\geq}r$ stating that \pzero can use this configuration (i.e., she
derives a non-zero value when getting to this configuration and this
non-zero value is 1) only if
she can ensure that the value she can get from this configuration
onwards meets the obligation.
It follows that, recursively, \pzero has to ensure that every obligation 
configuration satisfies the obligation requirement with plays in $\alpha$.
For configuration $v$, if $O(v)\neq \bot$ we call $v$ an
\emph{obligation} configuration and if $O(v)=\bot$ we call $v$ a
\emph{non-obligation} configuration.

As mentioned, the usual approach to defining values in games by
considering the measure of winning paths on a Markov chain and taking
the supremum of infimum of strategies of the respective players does
not work.
This is mainly for two reasons.
First, the definition of value over a Markov chain needs defining in
its own right.
Second, if a value is not achievable by a single strategy (which
is the case in Blackwell games with Borel objectives \emph{without}
obligations) the supremum over the value 
of strategies is not sufficient to capture the complexity of
obligations and (infinitely many) nested supremum (and infimum)
operators are required.
Intuitively, the value of a configuration in an
obligation game is the value in the modified
game where \pzero's objective is to either reach obligations she can
fulfil or never reach obligations and fulfil the Borel winning conditions.
If during this interaction a new obligation is met then this
obligation needs to be fulfilled in the same way.
If infinitely many obligations are met along a path, this path has to
be winning according to the Borel objective.
We present the formal definition of the value for the players in an
OBG through a reduction to a turn-based game similar to
Martin's proof that Blackwell games are determined \cite{Mar98}.

We generalize the function $O$ to apply to prefixes, where
$O(wv)=O(v)$ for every $wv\in V^+$ and similarly for $R$, and $\succ$.
Consider the game $\turn{G} = ((\hat{V},E),(V_0,V_1),\hat{\alpha})$, where
the components of $\turn{G}$ are given in
Figure~\ref{figure:components of turn G}.
\begin{figure*}[bt]
{\small
$$
\begin{array}{|l|}
\hline
\begin{array}{l l}
\hat{V}= & 
(V^+ \times \lorc{0}{1}) \quad \cup \quad 
(V^+ \times \lorc{0}{1} \times \set{\epsilon} )
\quad \cup \\
&
\left \{ (w,r,f) \left | 
\begin{array}{l}
w\in V^+, r\in \lorc{0}{1}, f:\succ(w) \rightarrow [0,1], \mbox{ and }
\exists d_0 \in \dist{A_0}~. \forall d_1 \in \dist{A_1}~. \\
\ds
\sum_{a_0\in A_0} \sum_{a_1\in A_1} 
\sum_{v'\in \succ(w)} d_0(a_0) \cdot d_1(a_1) \cdot R(v,a_0,a_1)(v')
\cdot f(v') > r
\end{array}
\right .
\right \} 
\end{array}
\\\hline
\begin{array}{l l}
V_0= & \set{(w,r) ~|~ \mbox{either $O(w)=\bot$ or $\exists
  r'.O(w)={>}r'$}} \quad \cup \quad 
  \set{(w,r,\epsilon) 
}
\end{array}
\\\hline
\begin{array}{l l}
 V_1 = & \set{(w,r) ~|~ \exists r'.O(w)={\geq}r'} \quad \cup \quad
 \set{(w,r,f)}
\end{array}
\\\hline
\begin{array}{l l l}
E = &
\set{((w,r),(w,r'',\epsilon)) ~|~ O(w)={\geq} r' \mbox{ and } 0<r''<r' }
\quad \cup \quad 
\set{((w,r,f),(w\cdot v',f(v'))) ~|~ f(v')>0}
& \cup  \\
& \set{((w,r),(w,r',f)) ~|~ \mbox{either }
O(w)=\bot \mbox{ and } r'=r 
\mbox{ or }
O(w)={>}r'  
} 
\quad \cup \quad
\set{((w,r,\epsilon),(w,r,f)) } 
\end{array}
\\\hline
\begin{array}{l l}
\hat{\alpha} = & \set{ p \in \hat{V}^\omega ~|~ p{\Downarrow}_{V} \in
  \alpha}, 
\mbox{where $p{\Downarrow}_{V}$ is the limit of the
    projection of $p$ on $V^+$.}
\end{array}
\\\hline
\end{array}
$$
\figurespace{-5mm}
\caption{\label{figure:components of turn G}Components of $\turn{G}$.}
}
\figurespace{-3mm}
\end{figure*}

There are three types of configurations.
Configurations of the form $(w,r)$, where $O(w)=\bot$,
are illustrated on the left in 
Figure~\ref{figure:structure of turn G}.
Such configurations are \pzero configurations, where she claims that
the value of prefix $w$ is \emph{more than} $r$.
From such configurations \pzero chooses a successor configuration
$(w,r,f)$, where $f$ is a function associating a value to every
successor of $w$ that proves that indeed the value at $w$ is greater
than $r$.
Configurations of the form $(w,r)$, where $O(w)= {>}r'$, are
illustrated in the middle in 
Figure~\ref{figure:structure of turn G}.
\begin{figure}[bt]
\begin{center}
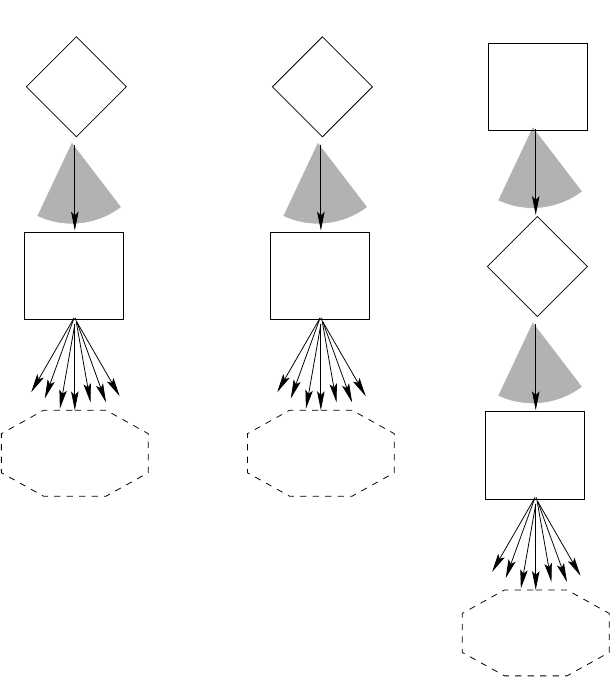 
\end{center}
\figurespace{-3mm}
\caption{\label{figure:structure of turn G}The structure of
  $\turn{G}$.
Diamonds are \pzero configurations and rectangles are \pone
configurations.
Shaded areas represent a continuum of edges, where every edge is
associated with an entry from the continuous domain written next to
the edge.
Fans of discrete edges represent finite choice, where every edge
is associated with a value from the domain written next to
the edge.
A dashed
octagon is either a \pzero or \pone configuration depending
on $O(v')$.}
\figurespace{-3mm}
\end{figure}
Such configurations are \pzero
configurations, where, ignoring the value $r$, she has to prove that
the value is greater than $r'$. Thus, she proceeds as above but for
the value $r'$ instead of $r$.
Configurations of the form $(w,r)$, where $O(w)= {\geq} r'$, are 
illustrated on the right in 
Figure~\ref{figure:structure of turn G}.
Such configurations are \pone
configurations, where, acknowledging that it may be impossible for
\pzero to achieve exactly $r'$ but possible to achieve every $r''<r'$,
\pone grants \pzero a concession and moves to a configuration
$(w,r'',\epsilon)$ from which, as above, \pzero chooses a successor
configuration $(w,r'',f)$.
Notice, that $\epsilon$ is used as a syntactic symbol signifying that
a concession has been granted, it is not a value.
Then from configurations of the form $(w,r,f)$, \pone chooses
which successor $v'$ of $w$ to follow and proceeds to 
$(w\cdot v',f(v'))$. 
Finally, we note that as $\alpha$ is a Borel set, then $\hat{\alpha}$
is also a Borel set.  

By Theorem~\ref{theorem:determinacy turn based games} for every prefix
$w$ and for every value $r\in \lorc{0}{1}$ from configuration $(w,r)$ 
in the game $\turn{G}$ either \pzero wins or else \pone wins.

\begin{lemma}
For every OBG $G$ and every prefix $w$, if \pzero wins from
$(w,r)$ in $\turn{G}$, she wins from every configuration $(w,r')$ for
$r'<r$. 
If \pone wins from $(w,r)$ in $\turn{G}$, she wins from every
configuration $(w,r')$ for $r'>r$.
\label{lemma:value closed downwards}
\end{lemma}

\newcommand{\proofoflemmavaluecloseddownwards}{
\begin{proof}
This can be done by reusing the strategy in $\turn{G}$.
Essentially, in order to show $r'<r$ \pzero can show $r$.
Dually, in order to show that $r'>r$ is infeasible it is enough to
show that $r$ is infeasible.
\end{proof}
}
\shorten{
\proofoflemmavaluecloseddownwards
}{}

\shorten{
So winning values for \pzero are downward closed and winning values
for \pone are upward closed.
It follows that there is a unique value below which
\pzero wins and above which \pone wins.
}
{
From Lemma~\ref{lemma:value closed downwards} it follows that there is
a unique value below which \pzero wins and above which \pone wins.
}

\begin{corollary}
For every OBG $G$ and every prefix $w$, there is a value
$s(G,w) \in [0,1]$
such that for every $r'<s(G,w)$ \pzero wins from $(w,r')$ in $\turn{G}$ and
for every $r''>s(G,w)$ \pone wins from $(w,r')$ in $\turn{G}$.
\end{corollary}

Notice that \pzero may or may not win from $s(G,w)$.
For a prefix $w$ we define the value of $w$ in $G$ as follows.
If $O(w)=\bot$ then the value of $w$ in $G$, denoted $\zval_0(G,w)$,
is $s(G,w)$. 
If $O(w)={\bowtie}r$ then $\zval_0(G,w)$ is 1 iff $s(G,w) >
r$ or $s(G,w)=r$ and ${\bowtie}={\geq}$ and it is 0 otherwise.

We now turn to the issue of determinacy.
In order to show that the value of \pzero and \pone sum to 1, we
define the dual game.
Dualization of a game consists of changing the roles of the two
players and switching the goal to the complement.
Here, the complementation of the goal is slightly more complicated
than usual.
Consider a game $G=(V,A_0,A_1,R,\cG)$, where $\cG=\pair{\varphi,O}$.
The dual game $\dual{G}=(V,A_1,A_0,\dual{R},\dual\cG)$,
where $\dual{R}(v,a_1,a_0)=R(v,a_0,a_1)$,
$\dual\cG=\pair{V^\omega \setminus \varphi,\dual{O}}$, and
$\dual{O}$ is defined below.
\begin{center}
$
\dual{O}(v) = 
\left \{
\begin{array}{l l}
\bot & \mbox{If } O(v)=\bot \\
{>}1-r & \mbox{If } O(v)={\geq}r \\
{\geq} 1-r & \mbox{If } O(v)={>}r
\end{array}
\right.
$
\end{center}
Intuitively, if in $G$ \pzero has the obligation to achieve more than
$r$ with the set $\varphi$, then the dual player (\pzero in
$\dual{G}$) has the obligation to achieve at least $1-r$ with the goal
set $V^\omega \setminus \varphi$.
Syntactically, $\dual{\dual{G}}=G$.
We use the dual game to define the value for \pone.
Formally, let $\zval_1(G,w)$ denote the value of $w$ in $\dual{G}$.
We prove that obligation games are determined by showing that the
sum of values of a prefix $w$ in $G$ and in $\dual{G}$ is 1.

\newcommand{\lemmavonneumannnessofvalues}{
\begin{lemma}
For every OBG and every prefix $w$, there are distributions
$d_0\in \dist{A_0}$ and $d_1\in \dist{A_1}$ such that
$$s(G,w) = 
\ds
\sum_{a_0\in A_0}
\ds
\sum_{a_1\in A_1} 
\ds
\sum_{v'\in\succ(v)}
d_0(a_0)\cdot d_1(a_1)\cdot R(w,a_0,a_1)(v') \cdot \zval_0(G,w\cdot v').
$$
Furthermore, for every $d'_0\in \dist{A_0}$ and $d'_1\in\dist{A_1}$
the following hold. 
$$
\begin{array}{l}
s(G,w) \leq 
\ds
\sum_{a_0\in A_0}
\ds
\sum_{a_1\in A_1} 
\ds
\sum_{v'\in\succ(v)}
d_0(a_0)\cdot d'_1(a_1)\cdot R(w,a_0,a_1)(v') \cdot \zval_0(G,w\cdot v').
\\
s(G,w) \geq 
\ds
\sum_{a_0\in A_0}
\ds
\sum_{a_1\in A_1} 
\ds
\sum_{v'\in\succ(v)}
d'_0(a_0)\cdot d_1(a_1)\cdot R(w,a_0,a_1)(v') \cdot \zval_0(G,w\cdot v').
\end{array}
$$
\label{lemma:von neumann-ness of values}
\end{lemma}
}

\newcommand{\proofoflemmavonneumannnessofvalues}{
\begin{proof}
Consider the values $\zval_0(G,w\cdot v')$ for $v'\in\succ(w)$.
By Von Neumann's minimax theorem \cite{Neu28} there is an $r\in [0,1]$
and optimal distributions $d_0\in\dist{A_0}$ and 
$d_1\in\dist{A_1}$ such that
$$
\ds
\sum_{a_0\in A_0}
\ds
\sum_{a_1\in A_1}
\ds
\sum_{v'\in \succ{w}}
d_0(a_0)\cdot d_1(a_1)\cdot R(w,a_0,a_1)(v')\zval_0(G,w\cdot v') = r
$$
\noindent
and for every $d'_0\in\dist{A_0}$ and every $d'_1\in\dist{A_1}$
we have 
$$
\begin{array}{l}
\ds
\sum_{a_0\in A_0}
\ds
\sum_{a_1\in A_1}
\ds
\sum_{v'\in \succ{w}}
d_0(a_0)\cdot d'_1(a_1)\cdot R(w,a_0,a_1)(v')\zval_0(G,w\cdot v') \geq r\\
\ds
\sum_{a_0\in A_0}
\ds
\sum_{a_1\in A_1}
\ds
\sum_{v'\in \succ{w}}
d'_0(a_0)\cdot d_1(a_1)\cdot R(w,a_0,a_1)(v')\zval_0(G,w\cdot v') \leq r
\end{array}
$$
\noindent
We have to show that $r=s(G,w)$.

Suppose that $r>s(G,w)$. Let $\delta=r-s(G,w)$ and consider a play
starting from $(w,s(G,w)+\frac{\delta}{2},f)$, where $f$ is the
function that associates $max(0,\zval_0(G,w\cdot v')-\frac{\delta}{4})$ to every
successor $v'$ of $w$.
Clearly, the minimax value for $f$ is at least
$r-\frac{\delta}{4}$, which is larger than
$s(G,w)+\frac{\delta}{2}$.
By definition of $\zval_0(G,w\cdot v')$, \pzero has a winning strategy from
$(w\cdot v',\zval_0(G,w\cdot v')-\frac{\delta}{4})$.
It follows that \pzero also wins from $(w,s(G,w)+\frac{\delta}{2})$
contradicting the definition of $s(G,w)$.

Suppose that $r<s(G,w)$. Let $\delta=s(G,w)-r$ and consider a play
starting from $(w,s(G,w)-\frac{\delta}{2})$.
Consider the configuration $(w,s(G,w)-\frac{\delta}{2})$.
In order to win, \pzero has to choose a successor configuration
$(w,s(G,w)-\frac{\delta}{2},f)$, where $f$ associates at least
$\zval_0(G,w\cdot v')+\frac{\delta}{2}$ with some successor $w\cdot v'$ of
$w$.
Then, by definition of $s(G,w)$, \pone wins from configuration
$(w\cdot v',f(v'))$ contradicting the definition of $s(G,w)$.
\end{proof}
}
\shorten{
\lemmavonneumannnessofvalues
\proofoflemmavonneumannnessofvalues

We are now ready to prove that value is well defined and that
obligation Blackwell games are determined.
}{}

\begin{theorem}
For All prefixes $w$ in an OBG $G$ we
have $\zval_0(G,w)+\zval_1(G,w)=1$.
\label{theorem:well defined values}
\end{theorem}

The proof of Theorem~\ref{theorem:well defined values}
is non-trivial. 
The proof requires Martin's determinacy proof style analysis of the uncountable
game $\turn{G}$, along with new subtleties 
(for example as shown in the example in 
Figure~\ref{figure:measure zero} that measure zero sets could play an
important role in values of obligation games). 

\newcommand{\proofoftheoremwelldefinedvalues}{
\begin{proof}
Note that the definition of our values is through
turn-based deterministic games, and thus relies on 
determinacy of turn-based deterministic games. 
In the present proof we do not explicitly rely on Borel
objectives, but the definition of values through turn-based
deterministic games requires determinacy for them
(and determinacy holds for turn-based deterministic games
with Borel objectives). 
More explicitly, our proof relies on determinacy for turn-based
deterministic games rather than Borel objectives.
The determinacy proof of Martin also relies on determinacy of
turn-based deterministic games.

We add a few comments for readers familiar with Martin's work.
We note that Martin considers a quantitative objectives that map plays
to payoffs in the range $[0,1]$  
while we consider whether \pzero is winning or not. This is
equivalent to restricting the payoffs to the range $\{0,1\}$. 
Furthermore, he uses the symbol for integration to represent the value
while we use the notation $\val{\cdot}{\cdot}$ and talk about winning.
The first part of the proof below corresponds to the construction of
the strategy for \pzero (p. 1570) and the proof of Lemma~1.1 in
Martin's paper.
The second part of the proof below corresponds to the construction of
the strategy for \pone (p. 1572) and the proof of Lemma~1.4.
The second half of Martin's paper considers various extensions of his
result. We do not touch upon similar subjects to his.

For a prefix $w$, let $S(G,w)$ denote the set of values $r$
such that \pzero wins from $(w,r)$ in $\turn{G}$.
\begin{compactitem}
\item[$\Rightarrow$]
We show that if $r\in S(G,w)$ then $1-r \notin S(\dual{G},w)$. 

Suppose that $r\in S(G,w)$.
That is, \pzero wins from $(w,r)$ in game $\turn{G}$.
We show that \pone wins from $(w,1-r)$ in $\turn{\dual{G}}$ proving
that $1-r \notin S(\dual{G},w)$.
Let $\sigma$ be the winning strategy of \pzero in $\turn{G}$. 
We now construct a winning strategy for \pone in $\turn{\dual{G}}$.
To distinguish between a prefix of a play in $G$ and prefixes in
$\turn{G}$ or $\turn{\dual{G}}$ we call the latter two \emph{paths}.
For a path $\dual{p}$ in $\turn{\dual{G}}$ we use the strategy
$\sigma$ to construct a path $p$ in $\turn{G}$ such that whenever
$\dual{p}$ ends in configuration $(w',t)$ then $p$ ends in configuration
$(w',r)$ such that $r+t \geq 1$.
Initially, we start from configuration $(w,1-r)$ in $\turn{\dual{G}}$
and from configuration $(w,r)$ in $\turn{G}$.
That is, both paths are of length one.

Suppose that the paths $p$ and $\dual{p}$ end in configurations 
$(w',r')$ and $(w',t')$, respectively, and that $t'+r'\geq 1$.
We have the following cases.
\begin{compactitem}
\item
Suppose that $O(w')=\bot$ then $(w',r')$ is a \pzero configuration in
$\turn{G}$ and $(w',t')$ is a \pzero configuration in $\turn{\dual{G}}$.
The winning strategy $\sigma$ instructs \pzero to choose some
configuration $(w',r',f)$ in $\turn{G}$.
Suppose that \pzero chooses the configuration $(w',t',f')$
in $\turn{\dual{G}}$.
By definition, there has to be a configuration $v'\in \succ(w')$ such
that $f(v')+f'(v')\geq 1$. 
%
%
%
We make \pone choose $(w'\cdot v',f(v'))$ in $\turn{G}$ and extend the
strategy $\pi$ of \pone in $\turn{\dual{G}}$ by choosing
$(w'\cdot v',f'(v'))$. 
\item
Suppose that $O(w')={>}r''$ in $G$. Then $O(w')={\geq}1-r''$ in
$\dual{G}$.
It follows that $(w',r')$ is a \pzero
configuration in $\turn{G}$ and $(w',t')$ is a \pone configuration in
$\turn{\dual{G}}$.
The winning strategy $\sigma$ instructs us to choose a configuration
$(w',r'',f)$ in $\turn{G}$.
From the minimax theorem \cite{Neu28}
it follows that there is a value $r'''>r''$ that is
attained for the optimal choice $d_0 \in \dist{A_0}$ such that
$$\ds\inf_{d_1\in \dist{A_1}} \hspace{-1mm} 
\left ( \hspace{-0.7mm}
\sum_{a_0\in A_0}\sum_{a_1\in A_1}\sum_{v'\in\succ(w')} 
\hspace{-4mm}d_0(a_0)\cdot d_1(a_1) \cdot R(w',a_0,a_1)(v') \cdot
f(v') \hspace{-1mm}\right ) 
$$
\noindent
is at least $r'''$. 
Let $\delta=r'''-r''$. Notice that $1-r''-\delta=1-r'''$.
Then, from configuration $(w',t')$ in $\turn{\dual{G}}$, \pone chooses
the successor configuration $(w',1-r''',\epsilon)$, in effect giving
up $\delta$ for \pzero's benefit.
Suppose that \pzero chooses the successor configuration
$(w',1-r''',f')$ in $\turn{\dual{G}}$.
As above, there has to be a successor $v'\in\succ(w')$ such that
$f'(v')+f(v')\geq 1$.
Then we make \pone choose $(w'\cdot v',f(v'))$ in $\turn{G}$ and extend
\pone's strategy in $\turn{\dual{G}}$ by the choice 
$(w\cdot v',f'(v'))$.
\item
Suppose that $O(w')={\geq} r''$ in $G$. Then $O(w')={>}1-r''$ in
$\dual{G}$.
It follows that $(w',r')$ is a \pone configuration in $\turn{G}$ and
$(w',t')$ is a \pzero configuration in $\turn{\dual{G}}$.
Suppose that \pzero chooses the successor configuration
$(w',1-r'',f')$ in $\turn{\dual{G}}$.
From the minimax theorem \cite{Neu28} it follows that there is a value
$r'''<r''$ that is attained for the optimal choice $d_1\in \dist{A_1}$
such that  
$$\ds\inf_{d_0\in \dist{A_0}} \hspace{-1mm}
\left ( \hspace{-0.7mm}
\sum_{a_0\in A_0}\sum_{a_1\in A_1}\sum_{v'\in\succ(w')}  \hspace{-1mm}
d_0(a_0)\cdot d_1(a_1) \cdot R(w',a_0,a_1)(v') \cdot f'(v') 
\hspace{-1mm}
\right )
$$
\noindent
is at least $1-r''$.
Let $\delta=r''-r'''$. Notice that $r''-\delta=r'''$.
Then, from configuration $(w',r')$ in $\turn{G}$,  we make \pone choose
the successor configuration $(w',r''',\epsilon)$, in effect giving
up $\delta$ for \pzero's benefit.
Now, \pzero's winning strategy in $\turn{G}$ instructs her to choose a
configuration $(w',r''',f)$.
As above, there has to be a successor $v'\in\succ(w')$ such that
$f(v')+f'(v')\geq 1$.
Then we make \pone choose $(w'\cdot v',f(v'))$ in $\turn{G}$ and extend
\pone's strategy in $\turn{\dual{G}}$ by the choice 
$(w'\cdot v',f'(v'))$.
\end{compactitem}
Consider the two infinite plays played in $\turn{G}$ and
$\turn{\dual{G}}$. 
Clearly, when projecting the two plays on the configurations in $V^+
\times \lorc{0}{1}$ that appear in them and then on the configurations
in $V^+$ we get exactly the same play.
By assumption $\sigma$ is a winning strategy for \pzero in
$\turn{G}$.
Hence, the limit of this projection is in $\alpha$ implying that the
strategy constructed for \pone in $\turn{\dual{G}}$ is indeed
winning.
\item[$\Leftarrow$]
We show that $1-s(G,w) \leq s(\dual{G},w)$.
Notice that if $s(G,w)=1$ then clearly, $1-s(G,w)\leq s(\dual{G},w)$.
We consider the case that $s(G,w)<1$.

By Lemma~\ref{lemma:value closed downwards} for every $r>s(G,w)$ we have \pone wins in $\turn{G}$
from $(w,r)$.
If $t=1-r$ then $t<1-s(G,w)$.
We show that \pzero wins from $(w,t)$ in $\turn{\dual{G}}$.

Consider some value $r > s(G,w)$ such that \pone wins from $(w,r)$.
We show that \pzero wins from $(w,1-r)$ in
$\turn{\dual{G}}$ by proving that $1-r\in S(\dual{G},w)$.
We use the difference between $1-r+s(G,w)$ and $1$ to give a winning
strategy for \pzero in $\turn{\dual{G}}$.
We use a winning strategy $\pi$ of \pone in $\turn{G}$ to
produce a winning strategy for \pzero in $\turn{\dual{G}}$.
For a path $\dual{p}$ in $\turn{\dual{G}}$ we use the winning strategy
$\pi$ of \pone in $\turn{G}$ to construct a path $p$ in $\turn{G}$
such that whenever $\dual{p}$ ends in configuration $(w',t')$ then $p$
ends in configuration $(w',r')$ such that $r'\geq s(G,w')$, \pone is
winning from $p$ using $\pi$, and $t' < 1-r'$.

Consider a configuration $(w,r)$ such that $r>s(G,w)$.
As $r>s(G,w)$ there is some $r>\tilde{r}>s(G,w)$ such that
\pone wins from $(w,\tilde{r})$.
Let $\pi$ be the winning strategy of \pone from $(w,\tilde{r})$. 
Initially, we start from configuration $(w,1-r)$ in $\turn{\dual{G}}$
and from configuration $(w,\tilde{r})$ in $\turn{G}$.
Clearly, $(w,\tilde{r})$ is winning for \pone, $\tilde{r}\geq s(G,w)$,
and $1-r < 1-\tilde{r}$.

Suppose that the two paths $p$ and $\dual{p}$ end in a configurations 
$(w',r')$ and $(w',t')$, respectively, and that $r'\geq s(G,w)$, \pone
wins from $p$ using $\pi$, and $t' < 1-r'$.
We have the following cases.
\begin{compactitem}
\item
Suppose that $O(w')=\bot$ then $(w',r')$ is a \pzero configuration in
$\turn{G}$ and $(w',t')$ is a \pzero configuration in
$\turn{\dual{G}}$.

For every location $v'\in\succ(w')$ let $u(v')$ be the following
value: 
$$
\inf \set{ 1, f(v') ~|~ (w'\cdot v'',r',f) \in \hat{V} \mbox{ and }
\pi(p \cdot (w'\cdot v'',r',f))=v'}
$$
That is, we consider all possible choices for \pzero from $(w',r')$.
Such a choice includes a function $f:\succ(w')\rightarrow[0,1]$.
Then, whenever the winning strategy of \pone chooses to proceed to
$v'$, we record the value promised by \pzero and take the infimum of
all these values.

By the minimax theorem there are $d_0\in\dist{A_0}$ and
$d_1\in\dist{A_1}$ such that 
$$
\ds\sum_{a_0\in A_0}\ds\sum_{a_1\in A_1}\ds\sum_{v'\in\succ(w')}
d_0(a_0)\cdot d_1(a_1)\cdot R(w',a_0,a_1)(v')\cdot u(v') =
\tilde{r}
$$
and $d_0$ and $d_1$ are the optimal distribution choices for both
players.
We show that $\tilde{r} \leq r'$.
Suppose by contradiction that $\tilde{r}>r'$.
Then, let $\epsilon=\frac{\tilde{r}-r'}{2}$ and consider the function 
$f(v'')=max(0,u(v'')-\epsilon)$.
Clearly, $(w',r',f)$ is a configuration in $\turn{G}$.
However, as $\pi$ is a winning strategy from $p$ the choice
$\pi(p\cdot (w',r',f))$ contradicts the definition of $u$.
So $\tilde{r}\leq r'$.

By assumption $t'<1-r'$. Let $\epsilon = 1-r'-t'$.
Consider now the function $f':\succ(w')\rightarrow [0,1]$ such that
$f'(v'')=1-u(v'')-\frac{\epsilon}{2}$.
The minimax value of $f'$ in $\turn{\dual{G}}$ is at least
$1-\tilde{r} -\frac{\epsilon}{2} \geq 1-r'-\frac{\epsilon}{2} > t'$.
Hence, $(w',t',f')$ is a configuration in $\turn{\dual{G}}$.

We extend {\pzero}'s strategy in $\turn{\dual{G}}$ by choosing
configuration $(w',t',f')$.
Then, \pone answers by choosing a successor $(w\cdot v',f'(v'))$.
Notice that it cannot be the case that $u(v')=1$.
Indeed, in such a case $f'(v')$ would be 0.
So the path $\dual{p}$ is extended by $(w',t',f')$ and then
$(w' \cdot v',f'(v'))$. 

We now turn our attention to extension of the path $p$.
By the choice of $u$, there is a function $f$ such that
$(w',r',f)$ is a configuration in $\turn{G}$, $\pi(p\cdot (w',r',f))$ is
$(w'\cdot v',f(v'))$, and either $f(v')=u(v')$ or $f(v')<
u(v')+\frac{\epsilon}{4}$. 
So we make \pzero choose in $\turn{G}$ the successor configuration
$(w',r',f)$.
Then, {\pone}'s winning strategy $\pi$ instructs her to choose
$(w'\cdot v',f(v'))$.

It follows that $f(v')\geq s(G,w'\cdot v')$. Otherwise, \pzero has a
winning strategy from $(w'\cdot v',f(v'))$ in contradiction with
{\pone}'s strategy $\pi$ being winning.
Furthermore, $\pi$ is winning from $p\cdot (w',r',f)\cdot
(w'\cdot v',f(v'))$.

%
%
Finally, as $f'(v')=1-u(v')-\frac{\epsilon}{2}$ and
$f(v')<u(v')+\frac{\epsilon}{4}$ we conclude that
$f'(v')<1-f(v')$.
\item
Suppose that $O(w')={>}r''$ in $G$. Then $O(w')={\geq}1-r''$ in
$\dual{G}$.
It follows that $(w',r')$ is a \pzero
configuration in $\turn{G}$ and $(w',t')$ is a \pone configuration in
$\turn{\dual{G}}$.
Suppose that \pone chooses the next configuration $(w',t'')$ in
$\turn{\dual{G}}$. 

Now, this is similar to the previous case, as we have to continue from
the configurations $(w',r'')$ in $\turn{G}$ and $(w',t'')$ in
$\turn{\dual{G}}$ that are both \pzero configurations and $t''<1-r''$.
\item
Suppose that $O(w')={\geq} r''$ in $G$. Then $O(w')={>}1-r''$ in
$\dual{G}$.
It follows that $(w',r'')$ is a \pone configuration in $\turn{G}$ and
$(w',1-r'')$ is a \pzero configuration in $\turn{\dual{G}}$.

The winning strategy $\pi$ instructs \pone to choose configuration
$(w',r''')$ such that $r'''<r''$.

As before, this is similar to the first case, as we have to continue
from the configurations $(w',r''')$ in $\turn{G}$ and $(w',1-r'')$ in
$\turn{\dual{G}}$ that are both $\pzero$ configurations and
$1-r''<1-r'''$.
\end{compactitem}
Consider the two infinite plays played in $\turn{G}$ and
$\turn{\dual{G}}$. 
Clearly, when projecting the two plays on the configurations in $V^+
\times \lorc{0}{1}$ that appear in them and then on the configurations
in $V^+$ we get exactly the same play.
By assumption $\pi$ is a winning strategy for \pzero in
$\turn{G}$.
Hence, this projection is not in $\alpha$ implying that the strategy
constructed for \pzero in $\turn{\dual{G}}$ is indeed winning.
\end{compactitem}
\end{proof}
}
\shorten{
\proofoftheoremwelldefinedvalues
}{}

\begin{corollary}
For every obligation Blackwell game $G$ and every prefix $w$
such that $O(w)\neq \bot$, $\zval_0(G,w) \in \set{0,1}$.
\label{corollary:value of obligation configurations}
\end{corollary}

\newcommand{\proofofcorollaryvalueofobligationconfigurations}{
\begin{proof}
Consider a configuration $v$ such that $O(w)={\bowtie}r$.
Then, by definition, the game $\turn{G}$ starting from configuration 
$(w,r')$ does not depend on the value $r'$.
It follows that either \pzero wins from $(w,r')$ for all $r'\in
\lorc{0}{1}$ or \pone wins from $(w,r')$ for all $r'\in\lorc{0}{1}$.
It follows that either $\zval_0(G,w)=1$ or $\zval_0(G,w)=0$.
\end{proof}
}
\shorten{
\proofofcorollaryvalueofobligationconfigurations
}{}

\mysection{Markov Chains with Obligations}
\label{section:simple definition for MC}

We show that for Markov chains the measure of an obligation objective
can be defined directly on the Markov chain.
This direct characterization is generalized later and is crucial for
the  algorithmic analysis for finite games with parity objectives.
We introduce the notion of a \emph{choice set}, a set of obligations
that \pzero can meet.
We then show that the definition of a value through a choice set and
the definition in Section~\ref{section:obligation blackwell games}
coincide. 

Consider a Markov chain $M=(S,P,L,s^\init)$.
Let $\cG=\pair{\alpha,O}$ be an obligation, where
$\alpha\subseteq S^\omega$ is a Borel set of infinite paths and
$O:S\rightarrow (\set{\geq , >} \times [0,1]) \cup \set{\bot}$ is the
obligation function.
We can think about such a Markov chain as an obligation Blackwell game
where $A_0$ and $A_1$ are singletons. 
Formally, $G_M=(S,\set{a},\set{a},R,\cG)$, where
$R(s,a,a)=P(s)$ for all $s\in S$.
As before, we are interested in sequences of locations, which
correspond to prefixes of plays in $G_M$.
Thus, we refer to them as prefixes also here.
Let $\widehat{S}$ denote the set of prefixes $s^{\init}\cdot
S^*$.
Let $\cO$ denote the set of locations $s \in S$ such that
$O(s)\neq \bot$ and $\widehat\cO$ prefixes $w\cdot s\in \widehat{S}$ such that
$s \in \cO$.
That is, $\cO$ is the set of locations with a non-empty
obligation and the set $\widehat\cO$ is the set of prefixes that end
in a location in $\cO$. 
We denote by $O(w)$ the obligation $O(s)$, where $w=w'\cdot s$. 
Let $\cN=S{\setminus}\cO$ denote the set of locations that have no
obligation and $\widehat\cN$ denote the set of prefixes
$\widehat{S}{\setminus}\widehat\cO$.
For a prefix $w$ a \emph{choice set} is $C_w\subseteq
\widehat\cO \cap (\set{w}\cdot S^+)$.
That is, it is a set of extensions of $w$ that have obligations.
For a prefix $w'\in \widehat{S}$ and a choice set $C_w$, an infinite
path $w'\cdot y$ is \emph{good} if either 
(a) $y=x\cdot z$, 
    $x\in \cN^* \cdot \cO$, and $w'\cdot x \in C_w$, 
or 
(b)
    $y \in \cN^\omega$ and $w'\cdot y \in \alpha$.
That is, either the first visit to $\cO$ after $w'$ is in $C_w$ or
$\cO$ is never visited and the infinite path is in $\alpha$.
Let $\beta_{C_w}^{w'}$ denote the set of good paths of $w'$ with choice set
$C_w$.
Given a choice set $C_w$ and a prefix $w'$, the measure of $\cG$
from $w'$ according to $C_w$ is:
$$\measure{M}\cG{C_w}{w'} =
\frac{\prob{M}{\beta_{C_w}^{w'}}}{\prob{M}{\set{w'}\cdot S^\omega}}.
$$
A choice set $C_w$ is \emph{good} if the following two conditions
hold:
\begin{compactitem}
\item
Every infinite path $\pi=s_0,s_1,\ldots$ in $M$ such that $\pi$ has
infinitely many prefixes in $C_w$ is in $\alpha$.
\item
For every sequence $w'\in C_w$ we have $\measure{M}\cG{C_w}{w'}
\bowtie r$, where $O(w')={\bowtie}r$.
\end{compactitem}
\noindent
Let $\cC_w$ denote the set of good choice sets for $w$.

Consider a Markov chain $M=(S,P,L,s^\init)$ and an
obligation $\cG=\pair{\alpha,O}$.
For prefix $w$ the pre-value of $w$ is
$$\preoval(M,\cG,w) = \ds\sup_{C \in \cC_w}\measure{M}\cG{C}{w}.$$
Finally, we define the value of $w$.
For a prefix $w$ such that $O(w)\neq \bot$ we define $\oval(M,\cG,w)$
to be $1$ if $\preoval(M,\cG,w) \bowtie r$, where 
$O(w)={\bowtie}r$, and $\oval(M,\cG,w)$ is $0$ otherwise.
For a prefix $w$ such that $O(w)=\bot$ we define $\oval(M,\cG,w)$ to be $\preoval(M,\cG,w)$.

We note that in a choice set $C$, if there is some prefix $w\notin C$
such that $w\in\widehat{O}$ then for every extension $w\cdot y$ of $w$
there is no point in including $w\cdot y$ in $C$.
Indeed, once a certain obligation is not included in $C$ all the
obligations that extend it are not important.
We restrict attention to choice sets that satisfy this restriction. 

We show that for every Markov chain and for every prefix the
above definition of value coincides with definition through Martin-like
reduction.

\begin{theorem}
For every Markov chain $M$, obligation
$\cG=\pair{\alpha,O}$, and prefix $w\in S^*$ we have
  $\oval(M,\cG,w)=\zval_0(G_M,w)$.
\label{theorem:simple value coincide MC}
\end{theorem}

The proof of Theorem~\ref{theorem:simple value coincide MC} requires
a refined analysis of a winning strategy in  
the uncountable game $\turn{M}$ obtained from a Markov chain $M$. 
Using this analysis we extract a witness choice set in $M$  
from a winning strategy in the uncountable game.

\newcommand{\proofoftheoremsimplevaluecoincideMC}{
\begin{proof}
We show that $\zval_0(G_M,w)\geq \oval(G,w)$.
\begin{compactitem}
\item
Fix $\epsilon>0$. 
We have to show that if there is a choice set $C$ that shows the value
$\oval(G,w)-\epsilon$ then \pzero can get the value
$\oval(G,w)-\epsilon$ in $\turn{G_M}$. 
This proof uses heavily Martin's proof of determinacy of Blackwell
games \cite{Mar98}. Fix a Markov chain $M=(S,P,L,s^{\init})$
for the rest of this proof.

Given a  Borel winning set $\beta\subseteq S^\omega$, Martin defines a
turn-based game $\widehat{G}_m$ that is slightly different to ours.
Formally,
$\widehat{G}_m= \pair{(S^+\times [0,1]) \cup (S^+\times
  F),E,\widehat{\beta}}$,
where $F$ is the set of functions $\set{f:S \rightarrow [0,1]}$ and 
$$
\begin{array}{r @{~=~} l}
E & \set{ ((w,v),(w,f)) ~|~ \sum_{s' \in S} P(s,s')\cdot f(s') \geq v}
\\
\multicolumn{2}{r}{
\rule{150pt}{0pt}
\cup \quad 
\set{((w,f),(w\cdot s,v)) ~|~ f(s)\geq v} }
.
\end{array}
$$
\noindent
For a set $P\subseteq S^+$, let $P\uparrow^w$ denote $P \cap
\set{w}\cdot S^+$, i.e., exactly all suffixes of $w$ in $P$.
Then, based on determinacy of $\widehat{G}_m$ and measurability of 
$\beta$ (since $\beta$ is Borel), Martin's proof shows that 
$\frac{\prob M{\beta\uparrow^w}}{\prob
  M{S^+\uparrow^w}}\geq v$ iff \pzero wins $\widehat{G}_m$
from every configuration $(w,v')$ for $v'<v$.
That is, \pzero announces the values she can derive from successors of
$w$ and \pone chooses a successor from which to show the value.
Finally, $\widehat{\beta}$ is the set of plays whose projection on $S^+$ has
limit in $\beta$.
The strategy of \pzero forces all infinite plays to be in
$\widehat{\beta}$.%
\footnote{
The main difference between the two games (except for no obligations
in Martin's version) is as follows.
In our game the value promised by \pzero is always slightly below the
real value.
Accordingly, we require that the weighted sum of values of the
successors be strictly larger than the promised value.
In Martin's version the weighted sum of values of successors may be
equivalent to the promised values (or larger).
}

We use Martin's result to show that whenever
$\measure{M}\cG{C}{w} \geq v$ then \pzero wins in $\turn{G_M}$
from $(w,v')$ for every $v'<v$.
Let $\delta=v-v'$.
As $\measure{M}\cG{C}{w} \geq v$, then according to
Martin's proof \pzero wins in $\widehat{G}_m$ from
$(s,v'+\frac{\delta}{2})$. 
We use \pzero's strategy in $\widehat{G}_m$ to win in $\turn{G_M}$.
As we play, we maintain the requirement in $\turn{G_m}$ always below the
requirement in $\widehat{G}_m$ by repeatedly dividing the gap between the
values in the two games by $2$.
It follows that in the $i$th round of playing the two games, the gap
between the values is $\frac{\delta}{2^i}$. 
Furthermore, as {\pzero}'s strategy in $\widehat{G}_m$ is winning it cannot
be the case that the play created passes through an obligation prefix
that is not in $C$ (indeed, all continuations from this point are
losing in $\widehat{G}_m$).
If on the other hand, a play passes through an obligation point that
is in $C$, then the correspondence between the game $\turn{G_M}$ and a
new instance of $\widehat{G}_m$ from the new obligation point is created.
Consider an obligation prefix $w'$ such that $O(w')={\geq} v'$.
By goodness of $C$, $\measure{M}\cG{C}{w'} \geq v'$.
In the game $\turn{G_M}$ \pone moves to a configuration
$(w',v'',\epsilon)$, where $v''<v'$.
Thus, we can use the same argument and use Martin's game $\widehat{G}_m$
to continue the strategy in $\turn{G_M}$.
Consider an obligation prefix $w'$ such that $O(w')={>}v'$.
By goodness of $C$, $\measure{M}\cG{C}{w'} > v'$.
Hence, there is some $v''$ such that $v'<v''<\measure{M}\cG{C}{w'}$ that can be used in Martin's game.
Finally, consider an infinite play in $\turn{G_M}$.
If the play visits $C$ infinitely often, then by $C$'s goodness, it is
winning for \pzero.
If the play visits $C$ finitely often, according to Martin's result,
the corresponding play is winning in $\widehat{G}_m$ implying that the
play is in $\alpha$.
\end{compactitem}

\noindent
In the other direction we show that $\oval(G,w)\geq \zval_0(G_M,w)$. 
\begin{compactitem}
\item
In the other direction, a winning strategy for \pzero in the game
$\turn{G_M}$ from $(w,v)$ induces a good choice set $C$.
We start by fixing the winning strategy $\sigma$ of \pzero. 
For the sake of this proof we assume that \pzero always plays by this
strategy $\sigma$.
Furthermore, from a prefix $w\in \widehat\cO$ such that $O(w)={\geq}
v$, we know that \pone can choose every successor $(w,v',\epsilon)$
for $v'<v$.
We restrict {\pone}s choices to those $v'$ such that $v'\geq 0$ and
$v'=v-\frac{1}{n}$ for some $n\in \mathbb{N}$.
Thus, when we say configuration is \emph{reachable} we 
mean under these choices of \pzero according to $\sigma$ and for
choices of \pone restricted in $\geq$-obligation configurations as
explained.

The definition of the choice set $C$ is quite involved as it 
has to take into account the infinitely many different strategies that
are involved in showing an obligation of the form $\geq r$
(corresponding to each of the choices $v'=v-\frac{1}{n}$).
We assume that the initial prefix $w$ is an obligation such
that $O(w)={\geq} v$ for the value that interests us $v$.
Indeed, this forces \pzero to be able to win the game
$(w,v',\epsilon)$ for every $v'=v-\frac{1}{n}$ for every
$n\in\mathbb{N}$.
The choice set we construct has to factor in these infinitely many
different strategies.
However, the same occurs whenever another obligation $w'$ that extends
$w$ is reached for which $O(w')={\geq}v''$ for some $v''$.
It follows, that the choice set we construct must include the same
construction for every obligation of the form ${\geq}v''$ that is
encountered in the game.
More formally, we have the following.

Assume that we start from prefix $(w,r')$ such that
$O(w)={\geq} v$.
Let $T$ denote the set of prefixes reachable
from $w$ excluding $w$ itself.
For every $w'=w\cdot s_1\cdots s_n \in T$, let $level(w')$ denote the
number of locations $s_i$ such that $s_i\in \cO$ and
$O(s_i)={\geq}r'$ for $1\leq i\leq n$ and $r'\in \lorc{0}{1}$.
That is, $level(w')$ is the number of ${\geq}$-obligations on the way
from $w$ to $w'$ excluding $w$ itself but including $w'$ (if appropriate).
For every $i\geq 1$, let $T_i\subseteq T$ denote the set $T_i=\set{w'
  ~|~ level(w')=i}$.
It follows that $T=\bigcup_{i\geq 1} T_i$ and for every $i$ and $j$ we
have $T_i\cap T_j=\emptyset$.
We now restrict attention (by induction) to a subset of the obligation
configurations that appear in $T$.

Consider an obligation $w'$ in $T_1$.
For every such obligations there is a minimal $n\in\mathbb{N}$ such
that $w'$ is reachable from $(w,v-\frac{1}{n},\epsilon)$. 
We call this the \emph{rank} of $w$, denoted $rank(w)$.
Consider an obligation $w'$ in $T_i$.
Let $s_1,\ldots, s_i$ be the $\geq$-obligations on the way from $w$ to
$w'$ and let $v_1,\ldots, v_i$ be the values of these obligations. 
Then, there is a minimal according to the lexicographic order
$(n_1,\ldots, n_i)$ such that $w'$ is reachable from $w$ by \pone
taking the choice $(s_j,v_j-\frac{1}{n_j},\epsilon)$ from $(s_j,v'_j)$
for the appropriate $v'_j$.
As before, we call this the \emph{rank} of $w'$, denoted $rank(w')$. 
We say that $w'$ is \emph{good} if for every $\geq$-obligation $w''$
on the path from $w$ to $w'$ we have that $rank(w'')$ is a prefix of
$rank(w')$. 
That is, whenever $w'$ is reachable through multiple choices of
\pone, we consider only the strategy \pone used from $w'$ for the
minimal choice of concession given on \emph{all} $\geq$-obligations on
the way to $w'$.
We say that $>$-obligation prefix $w''$ is \emph{good} if it
appears in $T$ and all $\geq$-obligation prefixes appearing on
the path from $w$ to $w''$ are good.
That is, if it appears as part of one of the same ``minimal''
strategies. 
Let $C=\set{w\in T\cap \widehat\cO ~|~ w \mbox{ is good}}$. 
We note that the definition of $C$ does not depend on $G_M$ being
derived from a Markov chain.
Indeed the same definition is used in the proof of
Theorem~\ref{theorem:simple value coincide SG}.

In order to show that $C$ is a good choice set we have to prove two
things.
First, that every path that visits infinitely many obligations in $C$
is in $\alpha$.
Second, that all obligations in $C$ are met.

For the first claim we note that 
an infinite sequence of prefixes in $C$ appears also in
$\turn{G_M}$ and from $\sigma$ being a winning strategy must be in
$\alpha$.
We have to show that the obligation of every $w'\in C$ is met.
However, for this we can use again Martin's reduction.
For every prefix $w\in C$, the strategy of \pzero in $\turn{G_M}$ can
be used to show a win in $\widehat{G}_m$ from $(w,v')$, where
$O(w)={\bowtie}r$ and either ${\bowtie}={\geq}$ and $v' \geq r$ or
${\bowtie}={>}$ and $v'>r$.
Essentially, {\pzero}'s strategy in $\turn{G_M}$ promises values for
each prefix visited in a play. These values are larger than the
values needed in $\widehat{G}_m$ and can be used to construct a strategy
in $\widehat{G}_m$.
If some prefix in $O$ is reached, then clearly the play must be
included in $C$ as there is a strategy of \pone that makes it
reachable.
Furthermore, every infinite play that includes infinitely many
prefixes in $C$ can be forced by \pone in $\turn{G_M}$ showing that it
is in $\alpha$.
It follows that in the Markov chain we have $\measure{M}\cG{C}{w}
\bowtie r$.
\end{compactitem}
\end{proof}
}
\shorten{
\proofoftheoremsimplevaluecoincideMC
}{}

The definition above uses the supremum over all good choice sets.
We show that there is a choice set that attains the supremum.

\begin{theorem}
For every Markov chain $M$, obligation
$\cG=\pair{\alpha,O}$, and prefix $w\in \widehat{S}$ there is a
choice set $C\in \cC$ such 
that 
$\measure{M}{\cG}{C}{w}=s(M,w)$.
\label{theorem:attained choice set MC}
\end{theorem}

\newcommand{\proofoftheoremattainedchoicesetMC}{
\begin{proof}
Fix a prefix $w$.
There are choice sets 
$\set{C_i}_{i \in   \mathbb{N}}$ such that 
$\measure{M}{\cG}{C_i}{w} \geq s(G,w)-\frac{1}{2^i}$.


Let $C'_0=C_0$.
Consider a set $C_{i+1}$.
Let 
$$C_{i+1}'=C_{i+1}{\setminus}\set{w ' \in C_{i+1} ~|~ \exists w'' \in
  C'_j \mbox{ for $j\leq i$ s.t. $w'=w''\cdot y$ for some $y$}}.$$ 
We set 
$C=\bigcup_{i=1}^\infty C'_i$.
We show that $C$ is a good choice set.

Consider an infinite path that visits infinitely many prefixes in $C$.
Clearly, all of the prefixes in $C$ belong to the same set $C_i$ for
some $i\in \mathbb{N}$.
Hence, by $C_i$ being a good choice set, the path is in $\alpha$.

Consider a point $w'\in C$.
Let $i$ be the minimal such that $w' \in C'_i$.
Then, for every extension $w''=w'\cdot y$ such that $w''\cdot y\in
C_i$ we have $w''\in C'_i$.
Indeed, as $w'$ does not have a prefix in $C'_j$ for all $j<i$, so is
the case for $w''$.
Then, as $\measure{M}\cG{C_i}{w'}$ satisfies the obligation of
$w'$ then so does $\measure{M}\cG{C}{w'}$.

We can show that $\measure{M}\cG{C}{w}=s(G,w)$.
Indeed, if it were smaller than $s(G,w)$ then there is an $i$ such
that $\measure{M}\cG{C_i}{w} > \measure{M}\cG{C}{w}$.
It must be the case that $C_i$ includes extensions of $w$ that are not
in $C$, contradicting the definition of $C$.
\end{proof}
}
\shorten{
\proofoftheoremattainedchoicesetMC
}{}

We can show that in some cases there is no one good choice set
that covers all possible prefixes.
Consider for example the Markov chain in Figure~\ref{figure:no global
  choice set}.
Suppose that the path in which $s_1$ appears infinitely often is not 
in $\alpha$.
Clearly, for every prefix $p=s_1\cdots s_1$ the pre-value of this
configuration is $1$.
Indeed, the choice set that includes exactly $p\cdot s_1$ proves that.
However, this choice set, establishes the value of $p\cdot s_1$ as
$\frac{2}{3}$, which is, as required, more than $\frac{1}{3}$.
A choice set that shows the values of \emph{all} prefixes
simultaneously, has to include all prefixes $s_1\cdots s_1$.
Thus, the infinite path $s_1\cdot s_1\cdots$ is visited infinitely
often by 
this choice set and it cannot be good.

\begin{figure}[bt]
\begin{center}
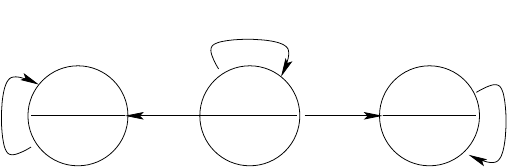 
\end{center}
\figurespace{-5mm}
\caption{\label{figure:no global choice set}An obligation Markov
  chain with no global good choice set.}
\figurespace{-4mm}
\end{figure}

\mysection{Finite Turn-Based Stochastic Parity Games with Obligations}
\label{section:simple definition for TB games}

We extend the results from obligation Markov chains to finite
turn-based stochastic parity games with obligations, and show that the
value function in such games has an alternate direct characterization
using choice sets. 
The direct characterization is crucial to present
algorithms to solve finite turn-based stochastic parity games with
obligations.
%
The simpler definition does not generalize to infinite games,
Blackwell games, or more general winning conditions.
In these more complicated games optimal strategies do not always exist.
For such games, we need a more elaborate construction that captures
the winning with $\epsilon$-optimal strategies (and non-existence of
optimal strategies), as done 
in Section~\ref{section:obligation blackwell games}.

We reuse the notation
$G=((V,E)$, $(V_0,V_1,V_p)$, $\kappa$, $\cG)$
for {\em turn-based stochastic parity games with obligations}.
Here, $\cG=\pair{\alpha,O}$ is a goal and $\alpha$ is derived
from a parity condition $c:V\rightarrow [0..k]$. 
Strategies are defined as before.
Given a prefix $w\in V^+$ and two strategies $\sigma$ and $\pi$, we
denote by $w(\sigma,\pi)$ the Markov chain obtained from $G$ by
using the strategies $\sigma$ and $\pi$ starting from prefix $w$.
Then, we define the value of \pzero in the prefix $w$ in the
game to be
$
\oval(G,w)=
\ds\sup_{\sigma\in\Sigma}
\ds\inf_{\pi\in\Pi}
\oval(w(\sigma,\pi),\cG,w).
$

We note that if we extend the definition of a measure of a choice set
so that bad choice sets give measure 0 for all configurations then the
following holds:
$$
\begin{array}{r}
\oval(G,w) =
\ds\sup_{\sigma\in\Sigma}\ds\inf_{\pi\in\Pi}\ds\sup_{C_w}
\measure{w(\sigma,\pi)}\cG{C_w}{w} \geq \quad \\
\multicolumn{1}{r}{
\geq \ds\sup_{\sigma\in\Sigma}\ds\sup_{C_w}\ds\inf_{\pi\in\Pi}
\measure{w(\sigma,\pi)}\cG{C_w}{w}}
\end{array}
$$
This follows from properties of supremum and infimum. 
In the proof of Theorem~\ref{theorem:simple value coincide SG}
we actually show that $\oval(G,w)\leq \zval_0(G,w) \leq \sup
\sup \inf (\cdots)$. Hence, the two are actually equivalent.
Formally, 
we show that for every finite turn-based stochastic parity game with
obligations and for every 
prefix the two values $\oval(G,w)$ and $\zval_0(G,w)$ coincide.

\begin{theorem}
For all finite turn-based stochastic obligation parity games
$G$ and prefix $w\in V^+$, $\oval(G,w)=\zval_0(G,w)$.
\label{theorem:simple value coincide SG}
\end{theorem}

\begin{figure}[bt]
\begin{center}
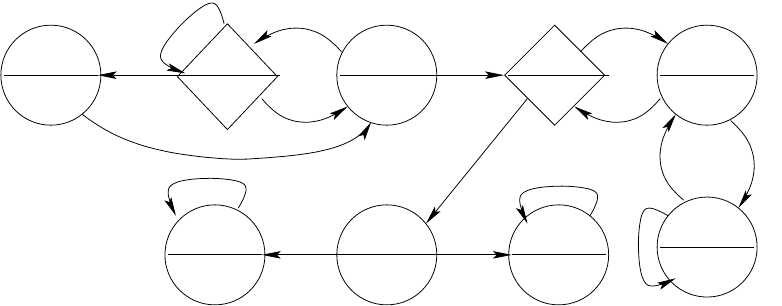
\end{center}
\figurespace{-5mm}
\caption{\label{figure:memoryfull}
  Finite turn-based stochastic parity game with obligations requiring
  memory. 
  Diamonds are \pzero configurations and circles are
  stochastic configurations. 
  Priorities in the range $[0..1]$ next to state define a
  parity acceptance condition. Only configurations $v_6$ and $v_8$
  have priority $0$.
}
\figurespace{-5mm}
\end{figure}

As another illustration of the notion of a choice set consider
the game in Figure~\ref{figure:memoryfull}.
In order to use $v_1$, \pzero has to win more than $\frac{1}{2}$ from
that configuration. 
Choosing the self loop from $v_3$ to itself or the edge from $v_3$ to
$v_1$ only makes things worse (though, each can be chosen a finite
number of times). 
So the only option from $v_3$ is to go
back to $v_2$ so that the probability of getting from $v_2$ to $v_4$
is $1$.
If from $v_4$, \pzero chooses to go to $v_7$ the value is
$\frac{1}{2}$ which does not satisfy the obligation of $v_1$.
Going from $v_4$ to $v_5$, on the other hand, and upon returning from
$v_5$ to $v_4$ proceeding to $v_7$ fulfills all obligations.
Indeed, the value for $v_1$ is $1$ as all paths eventually reach
$v_5$, and the value for $v_5$ is $\frac{3}{4}$ as the loop to itself
through $v_6$ is winning and the paths from $v_5$ to $v_4$ and then on
to $v_7$ have value $\frac{1}{4}$.
It follows that a possible choice set for this game is 
$
C=\{v_1,v_1v_2(v_3v_2)^*v_4v_5(v_6^+v_5)^*\}.
$
Indeed, \pzero has a strategy reaching from $v_1$ to
$v_1v_2(v_3v_2)^*v_4v_5$ with probability $1$.
She has a strategy from $v_1v_2(v_3v_2)^iv_4v_5(v_6^jv_5)^k$ to either
reach $v_1v_2(v_3v_2)^iv_4v_5(v_6^jv_5)^{k}v_6^+v_5$ or win the parity
objectives with probability $\frac{3}{4}$.
We note that \pzero uses its first visit to $v_4$ to go to $v_5$, in
order to boost the probability needed for the obligation of $v_1$, and
in subsequent visits goes to $v_7$.

We show that for a finite turn-based stochastic parity game $G$, from
a winning strategy in the uncountable game $\turn{G}$, we can
construct a regular witness choice set $C'$.
Using the regular witness and the fact that for finite turn-based
stochastic $\omega$-regular games finite-memory optimal
strategies exist we obtain 
Theorem~\ref{theorem:simple value coincide SG} .

\newcommand{\proofoftheoremsimplevaluecoincideSG}{
\begin{proof}
\noindent
As before, let $\cO=\{v \in V ~|~ O(v)\neq \bot\}$, $\cN=\{v\in V ~|~
O(v)=\bot\}$. 
Let $\cO_i= \{v \in \cO ~|~ c(v)=i\}$ and $\cN_i \{v\in \cN ~|~
c(v)=i\}$ be the obligation and non-obligation configurations with
priority $i$.
Similarly, let $\cO_{\geq i}=\{v\in \cO ~|~ c(v)\geq i\}$ and
$\cN_{\geq i}=\{v\in \cO ~|~ c(v)\geq i\}$. 

We show that $\oval(G,w)\geq \zval_0(G,w)$.

\begin{compactitem}
\item
Consider a winning strategy for \pzero in $\turn{G}$ that starts in
$(w,r)$.
We can extract from it a set $C$ of obligations that are used.
This is done just like in the proof of 
Theorem~\ref{theorem:simple value coincide MC}.
Clearly, every path that visits infinitely many prefixes in $C$ is in
$\alpha$ as it appears also in $\turn{G}$ as before.
We modify $C$ to a set $C'$ that we can show is good.
We construct $C'$ by induction on the number of obligations passed on
the way from $(w,r)$.

Consider a prefix $w'\in C \cup \{w\}$.
Let $ob(w')$ be the set of obligations $w''\in C$ such that
$w''=w'\cdot w'''$ and $w'''\in \cN^*\cdot \cO$.
That is, $ob(w')$ is the set of obligations directly reachable from
$w'$ without passing through other obligations.
Furthermore, annotate every prefix $w''$ in $ob(w')$ by the minimal
priority occurring in $w'''$, where $w''=w'\cdot w'''$ and $w''' \in
\cN^*\cdot \cO$.
For every prefix $w'\in C\cup \{w\}$ we define a set
$obs(w') \subseteq V \times [0..k]$, where $[0..k]$ are the priorities
of the parity condition.
Formally, $obs(w')$ is the set of pairs $(v',i')$ such that some $w''
\in ob(w')$ is annotated by $i'$ and the last configuration in $w''$
is $v'$.

Let $\cN_{i}$ and $\cO_i$ denote the configurations in $G$ whose priority
is $i$.
Let $\cN_{\geq i}$ and $\cO_{\geq i}$ denote the configurations in $G$
whose priority is at least $i$.

We now construct $C'$ by induction. We label every prefix $p'\in C'$
by a prefix $p\in C$ that is the \emph{reason} for inclusion of $p'$
in $C'$.
Consider the configuration $w$.
By construction, for every $(v,i)\in obs(w)$ there is a prefix
$w_{(v,i)} \in C$ such that $w_{(v,i)}\in ob(w)$ and $w_{(v,i)}$ ends in
$v$.
We add to $C'$ all the prefixes $w \cdot p$, where $p$ is in the
following set (restricted to those reachable from $w$):
$$
\bigcup_{(v,i) \in obs(w)} (\cN_{\geq i}^* \cdot \cN_i \cdot \cN_{\geq
  i}^* \cdot \{v\} \cup \cN_{\geq i}^* \cdot (\cO_i \cap \{v\}))$$
Furthermore, every prefix $w\cdot p$ for $p\in (\cN_{\geq i}^* \cdot \cN_i
\cdot \cN_{\geq i}^* \cdot \{v\} \cup \cN_{\geq i}^* \cdot (\cO_i \cap
\{v\}))$ is labeled by $w\cdot w_{(v,i)}$.
We now continue by induction.
Consider a prefix $p'\in C'$ that is labeled by prefix $p\in C$.
By construction, for every $(v,i)\in obs(p)$ there is a prefix
$p_{(v,i)} \in C$ such that $p_{(v,i)} \in ob(p)$ and $p_{(v,i)}$ ends
in $v$.
We add to $C'$ all the prefixes $p'\cdot p''$, where $p''$ is in the
following set (restricted to those reachable from $p'$):
$$
\bigcup_{(v,i) \in obs(p)} (\cN_{\geq i}^* \cdot \cN_i \cdot 
\cN_{\geq i}^* \cdot \{v\} \cup \cN_{\geq i}^* \cdot (\cO_i \cap \{v\}))$$
Furthermore, every prefix $p'\cdot p''$ for $p''\in (\cN_{\geq i}^* \cdot \cN_i
\cdot \cN_{\geq i}^* \cdot \{v\} \cup \cN_{\geq i}^* \cdot (\cO_i \cap
\{v\}))$ is labeled by $p\cdot p_{(v,i)}$.

This completes the construction of $C'$.
We have to show that $C'$ is a good choice set.
That is, every infinite path in $G$ that visits infinitely many
configurations in $C'$ is fair and the strategy of \pzero in $G$
establishes all the obligations posed by $C'$.

The fact that every infinite path in $G$ that visits infinitely many
prefixes in $C'$ is fair can be deduced by following the labels
in $C$ of prefixes in $C'$.
Consider such an infinite sequence of prefixes $p_0,p_1,\ldots$ in
$C'$ and their respective labels $w_0,w_1,\ldots$ from $C$.
By construction, the minimal priority visited in the extension of 
$p_i$ to $p_{i+1}$ is the minimal priority visited in the extension of
$w_i$ to $w_{i+1}$.
Furthermore, the sequence $w_0,w_1,\ldots$ corresponds to a path in
$G$ that visits infinitely many configurations in $C$.
As $C$ is obtained from $\turn{G}$, it follows that the limit of
$w_0,w_1,\ldots$ is fair.
That is, the limit of $w_0,w_1,\ldots$ satisfies the parity
objective.
We conclude that the limit of $p_0,p_1,\ldots$ is fair as well.

We now have to show that all obligations in $C'$ are met.
Consider a prefix $p'\in C'$ labeled by prefix $p\in C$.
Both $p'$ and $p$ end in the same configuration $v\in V$ of $G$.
It follows that the obligation $O(v)$ is fulfilled in $\turn{G}$.
Assume that $O(v)={\geq}r$.
It follows that \pzero wins in $\turn{G}$ from $(p,r')$ for every
$r'<r$. 
Recall the sets $ob(p) \subseteq C$ and $obs(p) \subseteq V \times
[0..k]$. 
The winning in $\turn{G}$ from $(p,r')$ for every $r'<r$ can be
translated to a win in $\hat{G}$ for the goal $ob(p)\cdot V^\omega$
for every $r''<r$, where $\hat{G}$ is the game obtained from $G$ by Martin's
reduction. 
Thus, the value of $ob(p)\cdot V^\omega \cup (\alpha \cap \cN^\omega)$ in $G$
is $r$. 
We now consider the following goal $\gamma$ in $G$:
$$\gamma = \bigcup_{(v,i) \in obs(p)} ((\cN_{\geq i}^*\cdot \cN_i \cdot
\cN_{\geq i}^* \cdot \{v\}) \cup (\cN_{\geq i}^* \cdot (\cO_i \cap \{v\}))) \cup 
(\alpha \cap \cN^\omega)$$
In particular, $\gamma$ contains at least all the extensions $p'$ such
that $p\cdot p'\in ob(p)$ as well as $\alpha\cap \cN^\omega$.
Furthermore, $\gamma$ can be translated to a parity goal
in $G$ by including a simple monitor for the minimal parity
encountered along the path.
As $reach(ob(p))\cup (\alpha \cap \cN^\omega) \subseteq \gamma$ it
follows that the value of $\gamma$ in $G$ is at least $r$.
However, values in finite turn-based stochastic parity games are
attained.
That is, there is a strategy for \pzero such that the value of
$\gamma$ according to this strategy is at least $r$.
It follows that by using this strategy \pzero can ensure the
obligation of $p$ in $G$.
The case that $O(v)={>}r$ is simpler, as \pzero wins directly from
$(v,r)$ in $\turn{G}$.

It follows that $C'$ is a good choice set and that \pzero has a
strategy that ensures that all obligations in $C'$ are met.
\end{compactitem}

\noindent
In the other direction we show that $\oval(G,w) \leq \zval_0(G,w)$.

\begin{compactitem}
\item
Suppose by way of contradiction that $r=\oval(G,w) > \zval_0(G,w)$.
Let $t<r$ be such that $t>\zval_0(G,w)$.
By definition of $\zval_0(G,w)$, \pone wins
in $\turn{\dual{G}}$ from $(w,1-t)$.
This winning strategy induces a good choice set
$T$ just like in the previous proofs.
Note that this set is good for \pone. 
Thus, every path that visits infinitely many configurations in
$T$ \emph{does not} satisfy the acceptance condition.
As in the other direction of the proof, the set $T$ can be extended to
a good choice set $T'$ such that \pone has a strategy to enforce all
obligations in $T'$ (in the dual game).
We show that $T'$ proves that the value of $w$ in $G$ cannot
be $r$.
Fix a strategy $\sigma$ of \pzero in $G$.
We show that the strategies of \pone in $G$ that enforce the
obligations in $T'$ (in the dual game) induce a strategy $\pi\in\Pi$
such that every choice set in $G_{\sigma,\pi}$ that shows the value
$r$ cannot be good.
Notice, that the set $\cO$ is the same in $G$ and $\dual{G}$.
Hence, there is a strategy $\pi$ for \pone in $G$ that achieves the
value $1-t$ for \pone.
Consider the strategy $\sigma$ and assume that it achieves the value
$r$ for \pzero in $G$ for some choice set $C$.
Then, as $1-t+r>1$, and $\cN^\omega \cap \alpha$ and
$\cN^\omega \cap \overline{\alpha}$ are disjoint, it follows
that $C$ and $T'$ have a non-empty intersection such that the strategy
$\sigma$ reaches $C\cap T'$.

We now proceed by induction.
Consider a configuration $w' \in C\cap T'$. 
There are two cases, either $O(w')={\geq} r'$ or $O(w')={>}r'$.
\begin{compactitem}
\item
Suppose that $O(w')={\geq} r'$.
In this case, the obligation of $w'$ in $\dual{G}$ is ${>}1-r'$.
It follows that there is some value $r''<r'$ such that the
value of $(T'\cdot V^\omega) \cup (\cN^\omega \cap \overline{\alpha})$ in
$G$ for \pone is $1-r''$.
Then, as $1-r''+r'>1$, and $\cN^\omega\cap \alpha$ and
$\cN^\omega \cap \overline{\alpha}$ are disjoint, it follows
that $C$ and $T'$ have a non-empty intersection such that the strategy
$\sigma$ reaches $C\cap T'$.
\item
Suppose that $O(w')={>}r'$.
In this case, there is a value $r''>r$ such that $\sigma$ must attain
the goal $(C \cdot V^\omega) \cup (\cN^\omega \cap \alpha)$ with
probability $r''$.
At the same time \pone can force the goal $(T'\cdot V^\omega)\cup
(\cN^\omega\cap \overline{\alpha})$ with probability $1-r'$.
As $r''+1-r'>1$ the sets $C$ and $T'$ have a non-empty
intersection such that the strategy $\sigma$ reaches $C\cap T'$.
\end{compactitem}
Continuing by induction we create a path that visits infinitely many
configurations in $T'$ and in $C$.
It follows that $C$ cannot be a good choice set.
\end{compactitem}
\end{proof}
}
\shorten{
\proofoftheoremsimplevaluecoincideSG
}{}

\begin{corollary}
  For every finite turn-based stochastic parity game with obligations
  $G$ and prefix $w \in V^+$, there are strategies $\sigma\in\Sigma$ and
  $\pi\in\Pi$ such that 
  $\oval(G,w)=\oval(w(\sigma,\pi),\cG,w)$.
  Furthermore, for every strategy $\sigma'\in\Sigma$ and $\pi'\in\Pi$
  we have 
  $
  \oval(w(\sigma',\pi),\cG,w) \leq 
  \oval(G,w) \leq
  \oval(w(\sigma,\pi'),\cG,w). 
  $
  \label{cor:values are attained}
\end{corollary}

\newcommand{\proofofcorollaryvaluesareattained}{
\begin{proof}
This follows from the proofs of 
Theorems~\ref{theorem:simple value coincide MC} and 
\ref{theorem:simple value coincide SG}.
Consider a configuration $w$.
Suppose that $\oval(G,w)=r$.
Then, for every $n$ there is a strategy $\sigma_n$ such that for every
$\pi\in\Pi$ the value $\oval(G_{\sigma_n,\pi}(w),\cG,w) \geq
r-\frac{1}{n}$.
Furthermore, there is a good choice set $C_n$ such that the goal
$(C_n \cdot V^\omega) \cup (\cN^\omega\cap \alpha)$ is enforced with
probability at least $r-\frac{1}{n}$.
As in the proof of Theorem~\ref{theorem:simple value coincide MC} the
different choice sets $\{C_n\}_{n >0}$ can be combined to a single
choice set $C$.
Furthermore, the choice set $C$ has a simple structure as in the proof
of Theorem~\ref{theorem:simple value coincide SG}.
It follows that \pzero can enforce the goal $(C\cdot V^\omega) \cup
(\cN^\omega \cap \alpha)$ with probability larger than
$r-\frac{1}{n}$ for every $n$.
As $G$ is finite it must be that $(C\cdot V^\omega)\cup
(\cN^\omega\cap \alpha)$ can be enforced with probability
$r$.

The proof that \pone also has an optimal strategy is similar.
\end{proof}
}
\shorten{
\proofofcorollaryvaluesareattained
}{}

\mysection{Algorithmic Analysis of Obligation Games}
\label{section:algorithms}

We give algorithms for solving obligation Blackwell games in two
cases.
First, in case that in every path in the game, the number of
transitions between an obligation configuration and a non-obligation
configuration is bounded.
In this case, we show that obligation Blackwell games can be reduced
to a sequence of turn-based stochastic games.
Second, in case that the game is finite and the winning condition is a
parity condition.
In this case, we give an exponential time algorithm for computing the
value of the game.

\mysubsection{Reduction to Stochastic Games}
\label{subsection:reduction}
Essentially, this is the solution adopted in \cite{HPW10} for solving
acceptance of uniform p-automata.
We partition the game to regions where there are no
transitions between obligation configurations and non-obligation
configurations. 
A region that consists only of non-obligation configurations can be
thought of as a stochastic game.
A region that consists only of obligation configurations can be
thought of as a turn-based (non-stochastic) game.
More formally, we have the following.

Consider an obligation Blackwell game $G=(V, A_0, A_1, R$, ${\cal G})$,
where ${\cal G}=\pair{\alpha,O}$.
We say that a configuration $v$ is \emph{pure} if for every
$a_0\in A_0$ and $a_1\in A_1$ we have $R(v,a_0,a_1)$ is pure.
We say that the game is \emph{uniform} if all the following
holds.
\begin{compactitem}
\item
There is a partition $\{V_i\}_{i\in \mathbb{N}}$ of $V$ such that for
every $i$ we have, either (i) for every $v\in V_i$ we have $O(v)=\bot$ or
(ii) for every $v\in V_i$ we have $O(v)\neq \bot$ or $v$ is pure.
\item
We say that $V_i\leq V_{i'}$ if there are some $v\in V_i$, $v'\in V_{i'}$,
$a_0\in A_0$, and $a_1\in A_1$ such that $R(v,a_0,a_1)(v')>0$.
The partition must also satisfy that every chain according to $\leq$ is
finite. 
\end{compactitem}

\begin{theorem}
The computation of the value of a uniform obligation Blackwell game
$G$ can be reduced to the solution of multiple Blackwell games.
\label{theorem:uniform blackwell games}
\end{theorem}

\newcommand{\proofoftheoremuniformblackwellgames}{
\begin{proof}
Let $\{V_i\}_{i\in \mathbb{N}}$ be the partition of the game $G$.
By assumption, consider a set $V_i$ such that there is no other set
$V_{i'}$ such that $V_i<V_{i'}$.
Consider a prefix $w=w'\cdot v$ such that $v\in V_i$.
Clearly, the extension of this prefix to a play in $G$ remains forever
in $V_i$.

Suppose that for all $v\in V_i$ we have $O(v)=\bot$.
Let $G'=(V^*,A_0,A_1,R,\alpha)$ be the game obtained from $G$ by
restricting attention to configurations reachable from $w$.
The game $G'$ is a normal Blackwell game and hence the value of every
configuration in $\{w\}\cdot V_i^*$ is well defined.

Suppose that for all $v\in V_i$ we have $O(v)\neq \bot$ or $v$ is pure.
Consider the turn-based game $G'=((V^* \cup V^*\times
2^V,E),(V^*,V^* \times 2^V),\alpha')$, where we
restrict $V^*$ to configurations reachable from $w$ and $E$ and
$\alpha'$ are as follows.
Consider a prefix $u'\cdot v'$ and a set $S\subseteq V$.
If $O(v')=\bot$, we say that
$S$ is \emph{possible} from $u'\cdot v'$ if there is $d_0\in {\cal
  D}(A_0)$ such that for all $d_1\in {\cal D}(A_1)$ we have 
\begin{equation}
\ds\sum_{a_0\in A_0} \ds\sum_{a_1\in A_1}\ds\sum_{v''\in S} 
R(v',a_0,a_1)(v'') \bowtie p,
\label{equation:minimax}
\end{equation}
\noindent
where $O(v')={\bowtie p}$.
If $v'$ is pure, we say that $S$ is \emph{possible} from $u'\cdot v'$
if there is $a_0\in A_0$ such that for all $a_1\in A_1$ the unique
configuration $v''$ such that $R(v,a_0,a_1)(v'')=1$ is in $S$.
Notice, that this is like considering a pure configuration as having
the obligation ${\geq}1$.
\begin{compactitem}
\item
$E=
\{(w'\cdot v',(w'\cdot v',S)) ~|~ \mbox{ $S$ possible from
  $w'\cdot v'$}\} \cup 
\{((w',S),(w'\cdot v')) ~|~ v'\in V\}$.
\item
$\alpha'$ includes all infinite paths such that the limit of their
projection on $V^*$ is in $\alpha$.
\end{compactitem}
This is in effect equivalent to the reduction to $\turn{G}$ when
restricted to $\{w\}\cdot V_i^*$.

Consider now a set $V_i$ and a configuration $w=w'\cdot v$ such that
$v\in V_i$.
Suppose, by induction, that for all configurations $u\cdot v'$ such
that $v'\in V_{i'}$ for $V_i<V_{i'}$ a value has already been computed.

If for every $v\in V_i$ we have $O(v)=\bot$, then a similar reduction
to a normal Blackwell game by plugging in the value of precomputed
configurations gives the value of all configurations in $\{w\}\cdot
V_i^*$.

If for every $v\in V_i$ we have $O(v)\neq \bot$, then a similar
reduction to a turn-based game can be done.
This time, value of precomputed configurations has to be combined in
the small minimax games as in Equation~\ref{equation:minimax}.
\end{proof}
}
\shorten{
\proofoftheoremuniformblackwellgames
}{}

\shorten{
We note that this is a ``meta''-algorithm.
Consider a uniform obligation Blackwell game $G$ and the partition
$V_1,\ldots, V_n$ showing that it is uniform.
Suppose that every $V_i$ reduces to a Blackwell game that can be
analyzed algorithmically.
Then, from Theorem~\ref{theorem:uniform blackwell games}, the game $G$
can be analyzed algorithmically.
}{
The proof of Theorem~\ref{theorem:uniform blackwell games} shows that
if each $V_i$ in the partition of $G$ is can be
analyzed algorithmically (as a Blackwell game) then so can $G$ (with
obligations). 
}

\mysubsection{Finite Turn-based Stochastic Obligation Parity Games}
\label{subsection:finite games}

We show that values in finite turn-based stochastic parity games with
obligations (\pog, for short) can be computed in exponential time and
decision problems regarding values lie in 
$\mbox{NP}{\cap}\mbox{co{-}NP}$.

We give a nondeterministic algorithm for finding a maximal (wrt to inclusion) 
choice set, which calls the computation of values in stochastic parity games
as a subroutine. 
Then, the value of a configuration in the game can be computed by
computing the value of reaching the choice set computed by the
algorithm or winning the parity condition without reaching other
obligations. 
By results of previous sections, dualization of the game gives the
value of the opponent. 
It follows that the decision regarding the value is also
in co-NP. 

We now give an algorithm that decides and computes values in $G$.
A \emph{dependency} for $v\in \cO$ is either $C_v=\bot$ or
$C_v\subseteq (\cO\times [0..k])$.
That is, $C_v$ is either \emph{undefined} or a (possibly empty) set of
pairs of 
obligation configurations annotated by priorities. 
A \emph{game dependency} is a set $\{C_v\}_{v\in \cO}$.
A game dependency is \emph{good} if the following conditions hold:
\begin{compactenum}
\item
If for some $v\in \cO$ we have $(v',i)\in C_v$ then 
$C_{v'}\neq \bot$. 
\item
For every infinite sequence $(v_0,i_0),(v_1,i_1),\ldots$ such that for
every $j$ we have $(v_{j+1},i_{j+1})\in C_{v_j}$ the minimal priority
occurring infinitely often in $i_0,i_1,\ldots$ is even.
\item
For every $v\in \cO$ such that $C_v\neq \bot$ we have
$\zval_0(G',v) \bowtie r$, where $O(v)={\bowtie}r$ and $G'$ is the
game $G$ considered as a turn-based stochastic game with the goal
$\gamma$:
$$
\bigcup_{(v',i) \in C_v} 
\left (
\begin{array}{l l}
(\cN_{\geq i}^*\cdot \cN_i \cdot \cN_{\geq i}^* \cdot (\cO_{\geq
  i}{\cap}\{v'\}) \cdot V^\omega) & \cup \\
(\cN_{\geq i}^* \cdot (\cO_i{\cap} \{v'\})\cdot V^\omega) & \cup \\
(\alpha \cap \cN^\omega)
\end{array}
\right )$$

\shorten{}{
Here $\cN_i=\cN\cap c^{-1}(i)$, $\cO_i=\cO\cap c^{-1}(i)$, 
$\cN_{\geq i}=\bigcup_{i'\geq i} \cN_{i'}$, and
$\cO_{\geq i}=\bigcup_{i'\geq i} \cO_{i'}$.
}
Informally, for an obligation $v$ with non-empty dependency, the dependency indeed shows that the obligation is met: \pzero can force 
(i)~winning the original winning condition while never reaching another 
obligation or
(ii)~reaching an obligation $v'$ that $v$ depends on, with $i$, the required 
parity, being the minimal visited along the way. 
\end{compactenum}

\begin{figure}[bt]
\begin{center}
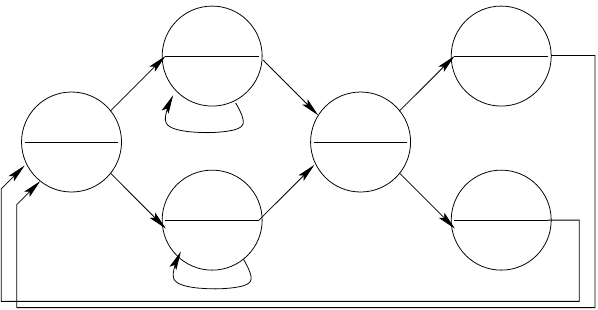
\end{center}
\figurespace{-5mm}
\caption{\label{figure:dependency}
  Illustration of a
  dependency. 
  Priorities in the range $[0..4]$ next to
  state define a parity acceptance condition.
}
\figurespace{-5mm}
\end{figure}

We illustrate the notion of a dependency using
Figure~\ref{figure:dependency}.
The obligation of $s_1$ is $\frac{3}{4}$.
The probability to reach $s_1$ from itself is 1.
However, the paths $(s_1s_2^+s_4s_6)^\omega$ have a minimal priority of
$1$ and are losing.
It follows that the only winning paths are
$(s_1s_2^+s_4s_5)^\omega$, $(s_1s_3^+s_4s_5)^\omega$, and
$(s_1s_3^+s_4s_6)^\omega$.
Thus, $s_1$ depends on reaching $s_1$ with minimal priority $0$
(through $s_3$) and on reaching $s_1$ with minimal priority $2$
(through $s_2$ and $s_5$).
This satisfies the three conditions as (1) $s_1$ has a defined
dependency, (2) every cycle visits either the minimal priority $0$ or
$2$, and (3)
the probability of reaching $s_1$ with minimal priority $0$ is
$\frac{1}{2}$, the probability of reaching $s_1$ with minimal priority
$2$ is $\frac{1}{4}$, and the probability of not reaching $s_1$ 
is $0$. So the total
probability is $\frac{3}{4}$, which fulfils the obligation of $s_1$.
Adding an obligation of ${\geq}\frac{1}{2}$ at $s_4$, changes the
dependency. 
Now, $s_1$ depends on reaching $s_4$ with priority $0$ or $2$ and
$s_4$ depends on reaching $s_1$ with priority $3$.
However, if the obligation of $s_4$ is set to ${>}\frac{1}{2}$, then
there is no good dependency.
Indeed, this would mean that \emph{whenever} $s_4$ is reached the path
through $s_6$ must be included.
Then, the path from $s_1$ through $s_2$ can not be part of the
dependency as this would create a cycle with minimum priority 1 and
the obligation of $s_1$ is no longer fulfilled.
The dependency for the game in Figure~\ref{figure:memoryfull} is $v_1$
depends on reaching $v_5$ with priority $1$ and $v_5$ depends on
reaching itself with priority $0$.
This is a good dependency as (a) the only cycle in it is $v_5$ reaching
itself with priority $0$ (b) from $v_1$ \pzero has a strategy that
ensures that $v_5$ is reached with probability $1$, and (c) from $v_5$
\pzero has a strategy that ensures that either $v_5$ is reached with
minimial priority $0$ encountered or getting to $v_8$ and staying
there (with no obligations on the way) with probability $\frac{3}{4}$.

The nondeterministic algorithm is as follows.
We guess a game dependency $\{C_v\}_{v\in \cO}$.
The size of $\{C_v\}_{v\in \cO}$ is polynomial in $|V|$. 
We check that $\{C_v\}_{v\in \cO}$ is good by doing the following.
First, checking that if $(v',i)\in C_v$ then $C_{v'}\neq \bot$
can be completed in polynomial time by scanning all the sets $C_v$.
Second, checking that all cycles induced by $C_v$ have a minimal even
parity in them can be completed in polynomial time by drawing the
graph of connections between the different configurations in $\cO$ for
which $C_v\neq\bot$ and searching for a cycle with minimal odd
priority. 
Third, ensuring that the values in the different turn-based
stochastic games fulfill the obligations can be achieved in
$\mbox{NP}{\cap}\mbox{co{-}NP}$ by
Theorem~\ref{theorem:complexity parity games}.  
Finally, consider the goal $\gamma'$:
$$
\gamma'=
(\cN^* \cdot \{v: V_c \neq \bot\}\cdot V^\omega) \cup 
(\alpha \cap \cN^\omega)
$$
We evaluate whether $\zval_0(G,w)\bowtie r$ by checking whether 
$\zval_0(G',w)\bowtie r$, where $G'$ is the turn-based stochastic game
obtained from $G$ by considering the goal $\gamma'$.
This can be checked in NP$\cap$co-NP.
To compute the value $\zval_0(G,w)$ we compute the value of $w$ in
$G'$. 
This can be computed in exponential time.
Notice that the values of $\gamma'$ in $G'$ correspond to the value
$s(G,w)$ and not $\zval_0(G,w)$.
For obligation configurations we must compare the result with the
required obligations.
If the obligation is met, the value $\zval_0(G,w)$ is $1$. 
Otherwise, it is $0$.
Overall, if \emph{all} the nondeterministic guesses are made up-front (i.e.,
the dependency \emph{and} the winning strategies in \emph{all} games)
then the global size of the witness is polynomial and all the checks
can be completed in polynomial time. 
Overall, the decision problem is in $\mbox{NP}{\cap}\mbox{co{-}NP}$,
and the values can be computed in exponential time.

We apply the algorithm on the example in
Figure~\ref{figure:dependency}.
As analyzed above, the dependency for $s_1$ is $(s_1,0)$ and
$(s_1,2)$. 
This proves that the value \emph{for all configurations} is $1$. 
Indeed, for every configuration in the game the probability of
reaching $s_1$ at least once is $1$.
Once $s_1$ is reached for the first time, the more complex reliance on
the choice set that is extracted from the dependency is required.
Applying the algorithm on the game in Figure~\ref{figure:memoryfull}
we see that the value of $v_1$, $v_2$, $v_3$, $v_4$, and $v_6$ is $1$
as from them \pzero can reach $v_5$ with probability $1$. The
pre-value of $v_5$ is $\frac{3}{4}$ as it reaches itself with
probability $\frac{1}{2}$ and wins the parity condition (reaching
$v_8$) with probability $\frac{1}{4}$. As this matches its obligation
its value is $1$.

Algorithm correctness follows from the following 
Lemmas.

\newcommand{\lemmanandproofofmemorylesswinninginturng}{
\begin{lemma}
There is a memoryless winning strategy in $\turn{G}$.
\label{lemma:memoryless winning in turnG}
\end{lemma}

\begin{proof}
According to the proof of 
Theorem~\ref{theorem:simple value coincide SG} from every obligation
used as part of the winning strategy in $\turn{G}$, there is a simple
goal that leads to the next frontier of used obligations.
Namely, given the sets $ob(p)$ and $obs(p)$ the goal is:
$$\gamma = \bigcup_{(v,i) \in obs(p)} ((\cN_{\geq i}^*\cdot \cN_i \cdot
\cN_{\geq i}^* \cdot \{v\}) \cup (\cN_{\geq i}^* \cdot (\cO_i \cap \{v\}))) \cup 
(\alpha \cap \cN^\omega)$$
The structure of this goal implies that there is a strategy with
memory linear in the number of priorities in the game that achieves an
optimal value for this goal.
However, an obligation configuration $v\in \cO$ may appear infinitely
often in $\turn{G}$, each time using a different strategy.

Consider now all the possible strategies for \pzero in $G$ with a
goal $\gamma$ as above with memory bounded by the number of
priorities.
Clearly, the number of such strategies is finite.
In particular, for every obligation configuration $v\in \cO$ there is
a finite number of strategies that are used in $\turn{G}$.
We now construct a finite parity game $G'$ based on these
strategies.
For every prefix $p\cdot v$ such that $v\in \cO$ used in $\turn{G}$ 
add a \pzero configuration $v$ to $G'$.
For every strategy $\sigma$ that is used from the prefix $p\cdot v$
in $\turn{G}$ we add a \pone configuration $\sigma$ to $G'$.
For every pair $(v',i)$ such that $v'\in \cO$ and $i$ is a priority such
that application of $\sigma$ from $p\cdot v$ reaches a configuration $v'$
with priority $i$ being the minimal visited along the way we add the
\pzero configuration $(\sigma,v',i)$ to $G'$.
We add edges to $G'$ as follows.
From configuration $v$ we add edges to all strategies $\sigma$ used
from $p\cdot v$ for some $p$.
From strategy $\sigma$ we add edges to all triplets $(\sigma,v',i)$.
From configuration $(\sigma,v',i)$ we add an edge to $v'$.
We set the priority of $(\sigma,v',i)$ to be $i$ and priorities of all
other configuration to be the maximal possible priority.

The game $G'$ is a finite parity game and we know that \pzero wins
$G'$ based on the combination of the winning strategies in $\turn{G}$.
It follows from 
Theorem~\ref{theorem:memoryless determinacy of parity games} 
that there is a memoryless winning strategy for \pzero in
$G'$.
However, a memoryless winning strategy in $G'$ induces a unique choice
of a strategy from every obligation configuration in $\turn{G}$
leading to a memoryless winning strategy in $\turn{G}$.
\end{proof}
}

\shorten{
\lemmanandproofofmemorylesswinninginturng
}{}

\begin{lemma}
An obligation configuration $v$ fulfills
$\zval_0(G,v)=1$ iff there is a good game dependency
$\{C_{v'}\}_{{v'}\in \cO}$ such that $C_v\neq \emptyset$.
\label{lemma:existence of dependency}
\end{lemma}

\newcommand{\proofoflemmaexistenceofdependency}{
\begin{proof}
The existence of a good game dependency clearly shows that the obligation
of $v$ can be met.

In the other direction, if the obligation of $v$ can be met, this
means that \pzero wins in $\turn{G}$ from $(v,r')$ for every $r'\in
(0,1]$. 
Furthermore, a choice set of a very particular form can be extracted
as in the proof of Theorem~\ref{theorem:simple value coincide SG}.
According to Lemma~\ref{lemma:memoryless winning in turnG} \pzero has
a memoryless winning strategy in $\turn{G}$.
We note further, that in $\turn{G}$, if an obligation configuration
$(w,r')$ occurs, the game below $(w,r')$ does not depend on the value
$r'$.
Thus, if two obligations $w\cdot v$ and $w'\cdot v$ and the
configurations $(w\cdot v,r)$ and $(w'\cdot v,r')$ occur in $\turn{G}$ 
the extension of the game below both is identical.
It follows, that the memoryless strategy behaves exactly the same from
all obligations $w\cdot v$ and $w'\cdot v$ for the same obligation
configuration $v$.
Then $obs(w\cdot v)=obs(w'\cdot v)$ for all $w,w'\in V^*$.
So for an obligation $v$ appearing in $\turn{G}$ we can use the set
$obs(w\cdot v)$ for some prefix $w$ as the dependency $C_v$. 
For every obligation $v'$ not appearing in $\turn{G}$ we set
$C_{v'}=\bot$. 
We have to show that this induces a good game dependency.
First, if we have $(v',i)\in C_v$ then it follows that some prefix
$w'\cdot v'$ is reachable from a prefix $w\cdot v$.
Thus, $C_v'$ must be defined.
Second, every cycle in $\{C_v\}_{v\in \cO}$ with minimal odd priority
corresponds to an infinite path in the good choice set with a minimal
odd priority, which is impossible.
Third, it must be the case that $\zval_0(G',v) \bowtie r$, where $G'$
is obtained from $G$ by considering the goal $\gamma$:
$$\gamma = \bigcup_{(v',i) \in C_v} ((\cN_{\geq i}^*\cdot \cN_i \cdot
\cN_{\geq i}^* \cdot \{v'\}\cdot V^\omega) 
\cup 
(\cN_{\geq i}^* \cdot (\cO_i\cap \{v\}) \cdot V^\omega))
\cup (\alpha \cap \cN^\omega)$$
Indeed, this is the exact construction of the choice set from
$\turn{G}$ as in the proof of Theorem~\ref{theorem:simple value
  coincide SG}, where it is proven that it is also good.
\end{proof}
}
\shorten{
\proofoflemmaexistenceofdependency
}{}

\begin{lemma}
For every configuration $v$, $\zval_0(G,v)=r$ iff 
there is a good game dependency
$\{C_{v'}\}_{{v'}\in \cO}$ such that $\zval(G',v)=r$, where $G'$ is
obtained from $G$ by considering the goal 
\shorten{
$\gamma$.
$$
\gamma = 
(\cN^* \cdot \{v: V_c \neq \bot\} \cdot V^\omega) \cup 
(\alpha \cap \cN^\omega)
$$
}{
$\gamma=
(\cN^* \cdot \{v: V_c \neq \bot\} \cdot V^\omega) \cup 
(\alpha \cap \cN^\omega)
$.}
\label{lemma:value computed}
\end{lemma}

\newcommand{\proofoflemmavaluecomputed}{
\begin{proof}
As before, if $\zval(G',v)=r$ then clearly $\zval_0(G,v) \geq r$.
In the other direction, we consider all the obligations appearing in
the choice set showing that $\zval_0(G,v)=r$.
According to the previous lemma, these obligations require a good game
dependency.
Finally, the value $\zval_0(G,v)$ is exactly the reachability of the
good choice set or winning parity without reaching obligations.
\end{proof}
}
\shorten{
\proofoflemmavaluecomputed
}{}

\begin{theorem}
For a {\pog} $G$ and a prefix $w\in V^+$, the values $\zval_0(G,w)$
and $\zval_1(G,w)$ can be 
computed in exponential time and whether $\zval_0(G,w)\bowtie r$ can be
decided in $\mbox{NP}{\cap}\mbox{co{-}NP}$.
\label{theorem:finite parity obligation games}
\end{theorem}

%

\mysection{p-Automata}
\label{section:paut}
In \cite{HPW10}, we defined uniform p-automata and showed that they
are a complete abstraction framework for pCTL.
Acceptance of Markov chains by uniform p-automata was defined through a
cumbersome and complicated reduction to a series of turn-based
stochastic parity games.
Here, using obligaton games, we give a clean definition of acceptance
by p-automata.
What's more, obligation games allow us to define acceptance by general
p-automata and remove the restriction of uniformity.
To simplify presentation we remove the notion of $*$-transitions (see
\cite{HPW10}). 

We assume familiarity with basic notions of trees and (alternating)
tree automata.
For set $T$, let $B^+(T)$ be the set of positive Boolean
\emph{formulas} generated from elements $t\in T$,
constants $\T$ and $\F$, and disjunctions and conjunctions:
%
\begin{equation}\label{equ:formulas}
\varphi,\psi ::= t\ \mid\ \T\ \mid\ \F\ \mid\ \varphi\lor\psi\ \mid\ \varphi\land\psi
\end{equation}
%
Formulas in $B^+(T)$ are finite even if $T$ is not.

For set $Q$, the set of states of a p-automaton, we define \emph{term}
sets $\dbr{Q}_{>}$ as follows.
$$
\dbr{Q}_{>} = \set{\bstate qp  \mid  q\in Q, {\bowtie} \in \set{\geq,
  >}, p\in [0,1]} 
$$

Intuitively, a state $q\in Q$ of a p-automaton and its transition
structure model a probabilistic path set. 
So  $\dbr{q}_{\bowtie p}$ holds in location $s$ if the measure of
paths that begin in $s$ and satisfy $q$ is $\bowtie p$. 

An element of $Q \cup \dbr{Q}_{>}$ is therefore either a state of the
p-automaton, or a term of the form $\dbr{q}_{\bowtie p}$.
Given $\varphi \in B^+(Q \cup \dbr{Q}_{>})$, its \emph{closure}
$\cl\varphi$ is  
the set of all subformulas of $\varphi$.
For a set $\Phi$ of formulas, let
$\cl\Phi=\bigcup_{\varphi\in\Phi}\cl\varphi$.

\begin{definition}
A p-automaton $A$ is a tuple
$\pair{\Sigma, Q, \delta, \varphi^\init,\alpha}$, 
where
$\Sigma$ is a finite input alphabet, 
$Q$ a set of states (not necessarily finite),
$\delta\colon Q \times \Sigma \rightarrow B^+(Q \cup \dbr{Q}_{>})$ the
transition function,
$\varphi^\init \in B^+(\dbr{Q}_{>})$ the initial
condition, and $\alpha$ a parity acceptance condition.
\label{def:p-aut main}
\end{definition}

In general, p-automata have {\em states}, 
Markov chains have {\em locations}, and 
games {\em configurations}.

\excludecomment{qestexample}
\begin{qestexample}
\begin{example}
Let $A=\pair{\ps{\set{\atom{a},\atom{b}}},\set{q_1,q_2},\delta, \dbr{q_1}_{\geq
    0.5},\set{q_2}}$ 
be a p-automaton 
 where $\delta$ is defined by
\vspace{-2.0mm}
{\small
$$
\begin{array}{l}
\delta(q_1,\set{\atom{a},\atom{b}}) =\delta(q_1,\set{\atom{a}}) = q_1 \lor
\dbr{q_2}_{\geq 0.5}  \\
\delta(q_2,\set{\atom{b}}) = \delta(q_2,\set{\atom{a},\atom{b}}) =
\dbr{q_2}_{\geq 0.5}\\ 
\delta(q_1,\set{}) = \delta(q_1,\set{\atom{b}}) = 
\delta(q_2,\set{}) = \delta(q_2,\set{\atom{a}}) = \F
\end{array}
$$
} 
Term $\dbr{q_2}_{\geq 0.5}$ represents the recursive property $\phi$, that atomic proposition $\atom{b}$ holds
at the location presently read by $q_2$, and that $\phi$
will hold with probability at least $0.5$ in the next locations.
State $q_1$ asserts that it is possible to get to a location that
satisfies $\dbr{q_2}_{\geq 0.5}$ along a path that satisfies atomic proposition $\atom{a}$.
The initial condition $\dbr{q_1}_{\geq 0.5}$ means the
set of paths satisfying $\atom{a}\U \phi$
has probability at least $0.5$. 
\label{example:a until b and b}
\end{example}
\end{qestexample}
%

For every $\ap$, p-automata 
$A=\pair{\ps\ap,Q,\delta,\varphi^\init,\alpha}$ have $\mc\ap$ as
set of inputs. 
For $M=(S,P,L,s^\init)\in \mc\ap$, we define whether $A$
accepts $M$ by a reduction to a turn-based stochastic parity game with
obligations.
The language of $A$ is
$\lang A = \set{M\in \mc\ap\mid A\hbox{
  accepts }M}$.

We construct a game
$G_{M,A} =
((V,E), (V_0,V_1,V_p), \kappa, {\cal G})$.
A configuration of $G_{M,A}$ corresponds to a subformula appearing in
the transition of $A$ and a location in $M$.
Configurations with a term of the form $\dbr{q}_{\bowtie p}$
correspond to obligations. 
All other configurations have no obligations.
The Markov chain is accepted if the configuration
$(\varphi^{\init},s^{\init})$ has value $1$ in
$G_{M,A}$. 

Formally, we define $G_{M,A}$ as follows.
Let $G_{M,A}=((V,E),(V_0,V_1,V_p),\kappa,{\cal G})$, where the
components of $G_{M,A}$ are as follows.
\begin{compactitem}
\item
$V=S \times \cl{\delta(Q,\Sigma)}$. 
\item
$V_0 = \{ (s,\psi_1\vee\psi_2) ~|~ s\in S \mbox{ and }\psi_1\vee\psi_2 \in
\cl{\delta(Q,\Sigma)}\}$.
\item
$V_1 = \{ (s,\psi_1\wedge\psi_2) ~|~ s\in S \mbox{ and
}\psi_2\wedge\psi_2\in \cl{\delta(Q,\Sigma)}\}$.
\item
$V_p = S \times (Q \cup \dbr{Q}_{>})$.
\item
The set of edges $E$ is defined as follows.
$$
\begin{array}{l@{=} l l}
E &
\set{((s,\varphi_1\wedge\varphi_2),(s,\varphi_i)) ~|~ i\in \set{1,2}}
& \cup \\
\multicolumn{2}{r}{
\set{((s,\varphi_1\vee\varphi_2),(s,\varphi_i)) ~|~ i\in\set{1,2}}}
&\cup \\
\multicolumn{2}{r}{\set{((s,q),(s',\delta(q,L(s)))) ~|~
s' \in \succ(s) }} &  \cup \\
\multicolumn{3}{r}{\set{((s,\dbr{q}_{\bowtie p}),(s',\delta(q,L(s)))) ~|~
s' \in \succ(s) }} \\
\end{array}
$$
\item
$\kappa((s,q),(s',\delta(q,L(s)))) =
\kappa((s,\dbr{q}_{\bowtie p}),(s',\delta(q,L(s)))) = P(s,s')$.
\item
${\cal G} = \pair{\tilde{\alpha},O}$, 
where 
\begin{compactitem}
\item
For $q\in Q$ and $p\in [0,1]$ we have $\tilde{\alpha}(s,q)=\alpha(q)$,
$\tilde{\alpha}(s,\dbr{q}_{\bowtie p})=\alpha(q)$.
For every other
configuration $c$ we set $\tilde{\alpha}(c)$ to the maximal possible
priority. 
\item
For $q\in Q$ and $p\in [0,1]$ we have $O(s,\dbr{q}_{\bowtie p}) =
{\bowtie}p$.
For every other configuration $c$, we have $O(c)=\bot$.
\end{compactitem}
\end{compactitem}

As obligation games are well defined it follows that it is well
defined whether a p-automaton accepts a Markov chain.
%

\begin{theorem}
Given a finite p-automaton $A$ and a finite Markov chain $M$, we can
decide whether $M\in \lang A$ in time exponential in the number of
states of $A$ and locations of $M$.
\label{theorem:complexity of acceptance}
\end{theorem}

\newcommand{\proofoftheoremcomplexityofacceptance}{
\begin{proof}
This follows from the polynomial construction of the finite-state
turn-based stochastic obligation parity game $G_{M,A}$ 
for the Markov chain $M$ and p-automata $A$, and Theorem~\ref{theorem:finite parity obligation games}.
\end{proof}
}
\shorten{
\proofoftheoremcomplexityofacceptance
}{}

The definition in \cite{HPW10} restricts attention to uniform
$p$-automata.
Such automata restrict the cycles in the transition graph of
p-automata.
We recall the definition of uniform p-automata.
In doing so, we differentiate states $q'$ appearing within a term in
$\dbr{Q}_{>}$ (bounded transition)  from $q'$ appearing ``free'' in the
transition of a state $q$ (unbounded transition).
In this way, a p-automaton $A = \pair {\Sigma,Q,\delta,\dots}$ 
determines a labeled, directed graph $G_{A}=\pair{Q',E,E_b,E_u}$:
\begin{equation*}
\begin{array}{lcl}
Q' &=& Q \cup \cl{\delta(Q,\Sigma)}\\
E &=&
\{(\varphi_1\land\varphi_2,\varphi_i),
  (\varphi_1\lor\varphi_2,\varphi_i) \mid \varphi_i \in Q'\setminus Q,\\
	&& \mbox{}\hfill  i\in\set{1,2}\} 
	\cup \set{ (q,\delta(q,\sigma)) \mid q\in Q, \sigma\in\Sigma } \\
E_u &=& \{ (\varphi\land q,q),(q\land\varphi,q),(\varphi\lor q,q),(q\lor\varphi,q) \mid \\
	&& \mbox{}\hfill \varphi\in Q', q \in Q \}\\
E_b &=& \set{ (\dbr{q}_{\bowtie p},q)  \mid  \dbr{q}_{\bowtie p} \in \dbr{Q}_{>} }
\end{array}
\end{equation*}
\noindent
Elements $(\varphi,q)\in E_u$ are \emph{unbounded} transitions;
elements $(\varphi,q)\in E_b$ are \emph{bounded} transitions;
and elements of $E$ are called \emph{simple} transitions. 
Note that $E$, $E_u$, and $E_b$ are pairwise disjoint.
Let $\varphi\preceq_A \tilde\varphi$ iff there is a finite path from
$\varphi$ to $\tilde\varphi$ in 
$E \cup E_b\cup E_u$.  Let $\equiv$ be $\preceq_A\cap \preceq_A^{-1}$
and $\eq \varphi$ the equivalence class of $\varphi$ with respect to
$\equiv$.
Each $\eq \varphi$ is an SCC in the directed graph $G_A$.

\begin{definition}{\cite{HPW10}}
A p-automaton $A$ is called  {\em uniform} if:
\begin{inparaitem}
\item[(a)]
For each cycle in $G_A$, its set of transitions is either in $E \cup
E_b$ or in $E \cup E_u$.
\item[(b)] There are only
finitely many equivalence classes $\eq\varphi$ with $\varphi\in Q \cup \cl{\delta(Q,\Sigma)}$.
\end{inparaitem}
\label{def:p-aut uniform}
\end{definition}

That is, $A$ is uniform, if the full subgraph of every equivalence class
in $\preceq_A$ contains only one type of non-simple transitions.
Also, all states $q'\in Q$ or formulas $\varphi$ occurring in 
$\delta(q,\sigma)$ for some $q\in Q$ and $\sigma\in\Sigma$ can be
classified as unbounded, bounded, or
simple~--~according to  
SCC $\eq q$.
\begin{qestexample}
\begin{example}
Figure~\ref{figure:graph and MC}{\rm(a)} depicts the graph $G_{A}$
for $A$ of Example~\ref{example:a until b and b}. 
p-Automaton $A$ is uniform: $\eq {q_1} =
\set{q_1,q_1\lor\dbr{q_2}_{\geq 0.5}}$ and $\eq {q_2} =
\set{q_2,\dbr{q_2}_{\geq 0.5}}$; 
in $\eq {q_1}$ there are no bounded edges, in $\eq {q_2}$ there
are no  unbounded edges; and
$G_A$ has no markings for $*$ or $\dstar$.
The SCC $\eq {\dbr{q_1}_{\geq 0.5}} = \set{\dbr{q_1}_{\geq 0.5}}$ is trivial.
In addition, $A$ is weak as $\alpha=\set{q_2}$. 
\end{example}
\end{qestexample}
Intuitively, the cycles in the structure of a uniform p-automaton $A$
take either no bounded edges or no unbounded edges.
Uniformity allowed to define acceptance for p-automata
through the solution of a sequence of stochastic games.

\begin{theorem}
p-automata (Definition~\ref{def:p-aut main}) extend the definition of uniform
p-automata (Definition~\ref{def:p-aut uniform}). 
\label{theorem:p-aut-extension}
\end{theorem}

\newcommand{\proofoftheorempautextension}{
\begin{proof}
The obligation Blackwell game resulting from the composition of a
uniform p-automaton with a Markov chain is a uniform turn-based
obligation game.
From Theorem~\ref{theorem:uniform blackwell games} it follows that its
value is obtained from solving a sequence of turn-based stochastic
games and turn-based games.
This gives rise to exactly the definition of acceptance through a
sequence of turn-based stochastic games and turn-based games as in
\cite{HPW10}.
\end{proof}
}
\shorten{
\proofoftheorempautextension
}{}

Closure under union and intersection is easy due to alternation 
(see~\cite{HPW10} for details).
Closure under complementation follows from our determinacy result for obligation games. 
%

\mysection{Conclusions and Future Work}
\label{section:conclusions}

We introduced obligations, a structural winning condition
that complements winning conditions on paths.
We show that Blackwell games with Borel objectives and obligations are
well defined.
We then present a simpler definition of value for Markov chains with
Borel objectives and obligations and for finite turn-based stochastic
parity games with obligations.
Based on the simpler definition we give algorithms for analyzing
finite turn-based stochastic parity games with obligations.
We then use games with obligations to define acceptance by
unrestricted p-automata, showing that the new definition generalizes a
previous definition for uniform p-automata.

This is one of the rare cases in games that arise in verification that
determinacy of games does not immediately follow from Martin's result
that Blackwell games with Borel objectives are determined.
The proof of determinacy uses elements from Martin's determinacy proof
but introduces new concepts that were not needed in that proof.
These new concepts are
required due to the more elaborate nature of games with obligations. 

Our work gives rise to many interesting questions.
For example, determining the complexity of other types of games such as 
Streett, Rabin, Muller, and quantitative games with obligations.



Finally, many questions regarding the theory of p-automata remain
open.
For instance,
understanding the different transition modes of such automata (i.e.,
alternation vs. nondeterminism vs. determinism) and conversions
between the different modes.
A related question is that of feasibility of algorithmic questions such
as emptiness of p-automata, which generalizes the satisfiability
problem of pCTL.


{

}

\shorten{}{
\clearpage
\onecolumn
\appendices

\input{appendix}


}

\end{document}

%% file: example.pdf_t
\begin{picture}(0,0)%
\includegraphics{example.pdf}%
\end{picture}%
\setlength{\unitlength}{1973sp}%
\begingroup\makeatletter\ifx\SetFigFont\undefined%
\gdef\SetFigFont#1#2#3#4#5{%
  \reset@font\fontsize{#1}{#2pt}%
  \fontfamily{#3}\fontseries{#4}\fontshape{#5}%
  \selectfont}%
\fi\endgroup%
\begin{picture}(4351,2995)(-187,-2948)
\put(301,-1036){\makebox(0,0)[b]{\smash{{\SetFigFont{6}{7.2}{\rmdefault}{\mddefault}{\updefault}{\color[rgb]{0,0,0}${>}\frac{2}{3}$}%
}}}}
\put(2026,-2311){\makebox(0,0)[b]{\smash{{\SetFigFont{6}{7.2}{\rmdefault}{\mddefault}{\updefault}{\color[rgb]{0,0,0}$s_4$}%
}}}}
\put(2026,-586){\makebox(0,0)[b]{\smash{{\SetFigFont{6}{7.2}{\rmdefault}{\mddefault}{\updefault}{\color[rgb]{0,0,0}$s_1$}%
}}}}
\put(1501,-1711){\makebox(0,0)[b]{\smash{{\SetFigFont{6}{7.2}{\rmdefault}{\mddefault}{\updefault}{\color[rgb]{0,0,0}$\frac{1}{3}$}%
}}}}
\put(2851,-136){\makebox(0,0)[b]{\smash{{\SetFigFont{6}{7.2}{\rmdefault}{\mddefault}{\updefault}{\color[rgb]{0,0,0}$\frac{1}{3}$}%
}}}}
\put(1051,-136){\makebox(0,0)[b]{\smash{{\SetFigFont{6}{7.2}{\rmdefault}{\mddefault}{\updefault}{\color[rgb]{0,0,0}$\frac{1}{3}$}%
}}}}
\put(3676,-586){\makebox(0,0)[b]{\smash{{\SetFigFont{6}{7.2}{\rmdefault}{\mddefault}{\updefault}{\color[rgb]{0,0,0}$s_3$}%
}}}}
\put(3676,-1036){\makebox(0,0)[b]{\smash{{\SetFigFont{6}{7.2}{\rmdefault}{\mddefault}{\updefault}{\color[rgb]{0,0,0}${\geq}\frac{1}{2}$}%
}}}}
\put(301,-586){\makebox(0,0)[b]{\smash{{\SetFigFont{6}{7.2}{\rmdefault}{\mddefault}{\updefault}{\color[rgb]{0,0,0}$s_2$}%
}}}}
\end{picture}%

%% file: measurezero.pdf_t
\begin{picture}(0,0)%
\includegraphics{measurezero.pdf}%
\end{picture}%
\setlength{\unitlength}{1973sp}%
\begingroup\makeatletter\ifx\SetFigFont\undefined%
\gdef\SetFigFont#1#2#3#4#5{%
  \reset@font\fontsize{#1}{#2pt}%
  \fontfamily{#3}\fontseries{#4}\fontshape{#5}%
  \selectfont}%
\fi\endgroup%
\begin{picture}(4610,1270)(-446,-1223)
\put(1201,-586){\makebox(0,0)[b]{\smash{{\SetFigFont{6}{7.2}{\rmdefault}{\mddefault}{\updefault}{\color[rgb]{0,0,0}$\frac{1}{2}$}%
}}}}
\put(301,-586){\makebox(0,0)[b]{\smash{{\SetFigFont{6}{7.2}{\rmdefault}{\mddefault}{\updefault}{\color[rgb]{0,0,0}$s_3$}%
}}}}
\put(2026,-586){\makebox(0,0)[b]{\smash{{\SetFigFont{6}{7.2}{\rmdefault}{\mddefault}{\updefault}{\color[rgb]{0,0,0}$s_1$}%
}}}}
\put(3676,-586){\makebox(0,0)[b]{\smash{{\SetFigFont{6}{7.2}{\rmdefault}{\mddefault}{\updefault}{\color[rgb]{0,0,0}$s_2$}%
}}}}
\put(2851,-136){\makebox(0,0)[b]{\smash{{\SetFigFont{6}{7.2}{\rmdefault}{\mddefault}{\updefault}{\color[rgb]{0,0,0}$\frac{1}{2}$}%
}}}}
\put(3676,-1036){\makebox(0,0)[b]{\smash{{\SetFigFont{6}{7.2}{\rmdefault}{\mddefault}{\updefault}{\color[rgb]{0,0,0}${>}\frac{1}{2}$}%
}}}}
\end{picture}%

%% file: turng.pdf_t
\begin{picture}(0,0)%
\includegraphics{turng.pdf}%
\end{picture}%
\setlength{\unitlength}{1973sp}%
\begingroup\makeatletter\ifx\SetFigFont\undefined%
\gdef\SetFigFont#1#2#3#4#5{%
  \reset@font\fontsize{#1}{#2pt}%
  \fontfamily{#3}\fontseries{#4}\fontshape{#5}%
  \selectfont}%
\fi\endgroup%
\begin{picture}(5859,6495)(1009,-6598)
\put(6376,-5011){\makebox(0,0)[lb]{\smash{{\SetFigFont{6}{7.2}{\rmdefault}{\mddefault}{\updefault}{\color[rgb]{0,0,0}\rotatebox{-55}{$\forall v':f(v')\neq 0$}}%
}}}}
\put(6151,-409){\makebox(0,0)[b]{\smash{{\SetFigFont{6}{7.2}{\rmdefault}{\mddefault}{\updefault}{\color[rgb]{0,0,0}$O(w)={\geq}r'$}%
}}}}
\put(6481,-3091){\makebox(0,0)[lb]{\smash{{\SetFigFont{6}{7.2}{\rmdefault}{\mddefault}{\updefault}{\color[rgb]{0,0,0}\rotatebox{-55}{$\forall f:V\rightarrow [0,1]$}}%
}}}}
\put(6451,-1711){\makebox(0,0)[lb]{\smash{{\SetFigFont{6}{7.2}{\rmdefault}{\mddefault}{\updefault}{\color[rgb]{0,0,0}\rotatebox{0}{$\forall r''<r'$}}%
}}}}
\put(6151,-2686){\makebox(0,0)[b]{\smash{{\SetFigFont{6}{7.2}{\rmdefault}{\mddefault}{\updefault}{\color[rgb]{0,0,0}$(w,r'',\epsilon)$}%
}}}}
\put(6151,-961){\makebox(0,0)[b]{\smash{{\SetFigFont{6}{7.2}{\rmdefault}{\mddefault}{\updefault}{\color[rgb]{0,0,0}$(w,r)$}%
}}}}
\put(6151,-4486){\makebox(0,0)[b]{\smash{{\SetFigFont{6}{7.2}{\rmdefault}{\mddefault}{\updefault}{\color[rgb]{0,0,0}$(w,r'',f)$}%
}}}}
\put(6151,-6286){\makebox(0,0)[b]{\smash{{\SetFigFont{6}{7.2}{\rmdefault}{\mddefault}{\updefault}{\color[rgb]{0,0,0}$(w\cdot v',f(v'))$}%
}}}}
\put(6256,-3241){\makebox(0,0)[lb]{\smash{{\SetFigFont{6}{7.2}{\rmdefault}{\mddefault}{\updefault}{\color[rgb]{0,0,0}\rotatebox{-55}{$\mbox{minimax}(f)>r''$}}%
}}}}
\put(4576,-286){\makebox(0,0)[b]{\smash{{\SetFigFont{6}{7.2}{\rmdefault}{\mddefault}{\updefault}{\color[rgb]{0,0,0}\rotatebox{-45}{$O(w)={>}r'$}}%
}}}}
\put(4456,-1411){\makebox(0,0)[lb]{\smash{{\SetFigFont{6}{7.2}{\rmdefault}{\mddefault}{\updefault}{\color[rgb]{0,0,0}\rotatebox{-55}{$\forall f:V\rightarrow [0,1]$}}%
}}}}
\put(4231,-1561){\makebox(0,0)[lb]{\smash{{\SetFigFont{6}{7.2}{\rmdefault}{\mddefault}{\updefault}{\color[rgb]{0,0,0}\rotatebox{-55}{$\mbox{minimax}(f)>r'$}}%
}}}}
\put(4088,-961){\makebox(0,0)[b]{\smash{{\SetFigFont{6}{7.2}{\rmdefault}{\mddefault}{\updefault}{\color[rgb]{0,0,0}$(w,r)$}%
}}}}
\put(4088,-4561){\makebox(0,0)[b]{\smash{{\SetFigFont{6}{7.2}{\rmdefault}{\mddefault}{\updefault}{\color[rgb]{0,0,0}$(w\cdot v',f(v'))$}%
}}}}
\put(4088,-2761){\makebox(0,0)[b]{\smash{{\SetFigFont{6}{7.2}{\rmdefault}{\mddefault}{\updefault}{\color[rgb]{0,0,0}$(w,r',f)$}%
}}}}
\put(4201,-3286){\makebox(0,0)[lb]{\smash{{\SetFigFont{6}{7.2}{\rmdefault}{\mddefault}{\updefault}{\color[rgb]{0,0,0}\rotatebox{-55}{$\forall v':f(v')\neq 0$}}%
}}}}
\put(1726,-961){\makebox(0,0)[b]{\smash{{\SetFigFont{6}{7.2}{\rmdefault}{\mddefault}{\updefault}{\color[rgb]{0,0,0}$(w,r)$}%
}}}}
\put(1726,-2761){\makebox(0,0)[b]{\smash{{\SetFigFont{6}{7.2}{\rmdefault}{\mddefault}{\updefault}{\color[rgb]{0,0,0}$(w,r,f)$}%
}}}}
\put(1726,-4561){\makebox(0,0)[b]{\smash{{\SetFigFont{6}{7.2}{\rmdefault}{\mddefault}{\updefault}{\color[rgb]{0,0,0}$(w\cdot v',f(v'))$}%
}}}}
\put(2101,-286){\makebox(0,0)[b]{\smash{{\SetFigFont{6}{7.2}{\rmdefault}{\mddefault}{\updefault}{\color[rgb]{0,0,0}\rotatebox{-45}{$O(w)=\bot$}}%
}}}}
\put(1876,-3286){\makebox(0,0)[lb]{\smash{{\SetFigFont{6}{7.2}{\rmdefault}{\mddefault}{\updefault}{\color[rgb]{0,0,0}\rotatebox{-55}{$\forall v':f(v')\neq 0$}}%
}}}}
\put(1801,-1561){\makebox(0,0)[lb]{\smash{{\SetFigFont{6}{7.2}{\rmdefault}{\mddefault}{\updefault}{\color[rgb]{0,0,0}\rotatebox{-55}{$\mbox{minimax}(f)>r$}}%
}}}}
\put(2026,-1411){\makebox(0,0)[lb]{\smash{{\SetFigFont{6}{7.2}{\rmdefault}{\mddefault}{\updefault}{\color[rgb]{0,0,0}\rotatebox{-55}{$\forall f:V\rightarrow [0,1]$}}%
}}}}
\end{picture}%

%% file: choice.pdf_t
\begin{picture}(0,0)%
\includegraphics{choice.pdf}%
\end{picture}%
\setlength{\unitlength}{1973sp}%
\begingroup\makeatletter\ifx\SetFigFont\undefined%
\gdef\SetFigFont#1#2#3#4#5{%
  \reset@font\fontsize{#1}{#2pt}%
  \fontfamily{#3}\fontseries{#4}\fontshape{#5}%
  \selectfont}%
\fi\endgroup%
\begin{picture}(4869,1570)(-371,-1223)
\put(2026,164){\makebox(0,0)[b]{\smash{{\SetFigFont{6}{7.2}{\rmdefault}{\mddefault}{\updefault}{\color[rgb]{0,0,0}$\frac{1}{3}$}%
}}}}
\put(2026,-586){\makebox(0,0)[b]{\smash{{\SetFigFont{6}{7.2}{\rmdefault}{\mddefault}{\updefault}{\color[rgb]{0,0,0}$s_1$}%
}}}}
\put(2026,-961){\makebox(0,0)[b]{\smash{{\SetFigFont{6}{7.2}{\rmdefault}{\mddefault}{\updefault}{\color[rgb]{0,0,0}${>}\frac{1}{3}$}%
}}}}
\put(376,-586){\makebox(0,0)[b]{\smash{{\SetFigFont{6}{7.2}{\rmdefault}{\mddefault}{\updefault}{\color[rgb]{0,0,0}$s_2$}%
}}}}
\put(3751,-586){\makebox(0,0)[b]{\smash{{\SetFigFont{6}{7.2}{\rmdefault}{\mddefault}{\updefault}{\color[rgb]{0,0,0}$s_3$}%
}}}}
\put(2926,-586){\makebox(0,0)[b]{\smash{{\SetFigFont{6}{7.2}{\rmdefault}{\mddefault}{\updefault}{\color[rgb]{0,0,0}$\frac{1}{3}$}%
}}}}
\put(1201,-586){\makebox(0,0)[b]{\smash{{\SetFigFont{6}{7.2}{\rmdefault}{\mddefault}{\updefault}{\color[rgb]{0,0,0}$\frac{1}{3}$}%
}}}}
\end{picture}%

%% file: memoryfull.pdf_t
\begin{picture}(0,0)%
\includegraphics{memoryfull.pdf}%
\end{picture}%
\setlength{\unitlength}{1973sp}%
\begingroup\makeatletter\ifx\SetFigFont\undefined%
\gdef\SetFigFont#1#2#3#4#5{%
  \reset@font\fontsize{#1}{#2pt}%
  \fontfamily{#3}\fontseries{#4}\fontshape{#5}%
  \selectfont}%
\fi\endgroup%
\begin{picture}(7478,2917)(3788,-4598)
\put(10276,-3286){\makebox(0,0)[b]{\smash{{\SetFigFont{6}{7.2}{\rmdefault}{\mddefault}{\updefault}{\color[rgb]{0,0,0}$\frac{1}{2}$}%
}}}}
\put(8251,-2236){\makebox(0,0)[b]{\smash{{\SetFigFont{6}{7.2}{\rmdefault}{\mddefault}{\updefault}{\color[rgb]{0,0,0}$\frac{1}{2}$}%
}}}}
\put(8251,-3961){\makebox(0,0)[b]{\smash{{\SetFigFont{6}{7.2}{\rmdefault}{\mddefault}{\updefault}{\color[rgb]{0,0,0}$\frac{1}{2}$}%
}}}}
\put(9826,-4111){\makebox(0,0)[b]{\smash{{\SetFigFont{6}{7.2}{\rmdefault}{\mddefault}{\updefault}{\color[rgb]{0,0,0}$\frac{1}{2}$}%
}}}}
\put(6976,-4486){\makebox(0,0)[b]{\smash{{\SetFigFont{6}{7.2}{\rmdefault}{\mddefault}{\updefault}{\color[rgb]{0,0,0}$1$}%
}}}}
\put(8626,-4486){\makebox(0,0)[b]{\smash{{\SetFigFont{6}{7.2}{\rmdefault}{\mddefault}{\updefault}{\color[rgb]{0,0,0}$0$}%
}}}}
\put(5326,-4486){\makebox(0,0)[b]{\smash{{\SetFigFont{6}{7.2}{\rmdefault}{\mddefault}{\updefault}{\color[rgb]{0,0,0}$1$}%
}}}}
\put(8026,-2836){\makebox(0,0)[b]{\smash{{\SetFigFont{6}{7.2}{\rmdefault}{\mddefault}{\updefault}{\color[rgb]{0,0,0}$1$}%
}}}}
\put(6001,-2236){\makebox(0,0)[b]{\smash{{\SetFigFont{6}{7.2}{\rmdefault}{\mddefault}{\updefault}{\color[rgb]{0,0,0}$v_3$}%
}}}}
\put(6751,-3961){\makebox(0,0)[b]{\smash{{\SetFigFont{6}{7.2}{\rmdefault}{\mddefault}{\updefault}{\color[rgb]{0,0,0}$\frac{1}{2}$}%
}}}}
\put(6751,-2236){\makebox(0,0)[b]{\smash{{\SetFigFont{6}{7.2}{\rmdefault}{\mddefault}{\updefault}{\color[rgb]{0,0,0}$\frac{1}{2}$}%
}}}}
\put(4276,-2236){\makebox(0,0)[b]{\smash{{\SetFigFont{6}{7.2}{\rmdefault}{\mddefault}{\updefault}{\color[rgb]{0,0,0}$v_1$}%
}}}}
\put(3901,-2011){\makebox(0,0)[b]{\smash{{\SetFigFont{6}{7.2}{\rmdefault}{\mddefault}{\updefault}{\color[rgb]{0,0,0}$1$}%
}}}}
\put(4276,-2686){\makebox(0,0)[b]{\smash{{\SetFigFont{6}{7.2}{\rmdefault}{\mddefault}{\updefault}{\color[rgb]{0,0,0}${>}\frac{1}{2}$}%
}}}}
\put(5626,-2836){\makebox(0,0)[b]{\smash{{\SetFigFont{6}{7.2}{\rmdefault}{\mddefault}{\updefault}{\color[rgb]{0,0,0}$1$}%
}}}}
\put(5851,-3961){\makebox(0,0)[b]{\smash{{\SetFigFont{6}{7.2}{\rmdefault}{\mddefault}{\updefault}{\color[rgb]{0,0,0}$v_9$}%
}}}}
\put(7501,-2236){\makebox(0,0)[b]{\smash{{\SetFigFont{6}{7.2}{\rmdefault}{\mddefault}{\updefault}{\color[rgb]{0,0,0}$v_2$}%
}}}}
\put(7501,-3961){\makebox(0,0)[b]{\smash{{\SetFigFont{6}{7.2}{\rmdefault}{\mddefault}{\updefault}{\color[rgb]{0,0,0}$v_7$}%
}}}}
\put(9076,-2236){\makebox(0,0)[b]{\smash{{\SetFigFont{6}{7.2}{\rmdefault}{\mddefault}{\updefault}{\color[rgb]{0,0,0}$v_4$}%
}}}}
\put(9151,-1861){\makebox(0,0)[b]{\smash{{\SetFigFont{6}{7.2}{\rmdefault}{\mddefault}{\updefault}{\color[rgb]{0,0,0}$1$}%
}}}}
\put(9151,-3961){\makebox(0,0)[b]{\smash{{\SetFigFont{6}{7.2}{\rmdefault}{\mddefault}{\updefault}{\color[rgb]{0,0,0}$v_8$}%
}}}}
\put(10576,-2236){\makebox(0,0)[b]{\smash{{\SetFigFont{6}{7.2}{\rmdefault}{\mddefault}{\updefault}{\color[rgb]{0,0,0}$v_5$}%
}}}}
\put(10576,-3886){\makebox(0,0)[b]{\smash{{\SetFigFont{6}{7.2}{\rmdefault}{\mddefault}{\updefault}{\color[rgb]{0,0,0}$v_6$}%
}}}}
\put(10576,-2686){\makebox(0,0)[b]{\smash{{\SetFigFont{6}{7.2}{\rmdefault}{\mddefault}{\updefault}{\color[rgb]{0,0,0}${\geq}\frac{3}{4}$}%
}}}}
\put(11251,-4111){\makebox(0,0)[b]{\smash{{\SetFigFont{6}{7.2}{\rmdefault}{\mddefault}{\updefault}{\color[rgb]{0,0,0}$0$}%
}}}}
\put(10876,-3286){\makebox(0,0)[b]{\smash{{\SetFigFont{6}{7.2}{\rmdefault}{\mddefault}{\updefault}{\color[rgb]{0,0,0}$\frac{1}{2}$}%
}}}}
\put(9676,-2686){\makebox(0,0)[b]{\smash{{\SetFigFont{6}{7.2}{\rmdefault}{\mddefault}{\updefault}{\color[rgb]{0,0,0}$\frac{1}{2}$}%
}}}}
\put(11026,-2011){\makebox(0,0)[b]{\smash{{\SetFigFont{6}{7.2}{\rmdefault}{\mddefault}{\updefault}{\color[rgb]{0,0,0}$1$}%
}}}}
\end{picture}%

%% file: dependency.pdf_t
\begin{picture}(0,0)%
\includegraphics{dependency.pdf}%
\end{picture}%
\setlength{\unitlength}{1973sp}%
\begingroup\makeatletter\ifx\SetFigFont\undefined%
\gdef\SetFigFont#1#2#3#4#5{%
  \reset@font\fontsize{#1}{#2pt}%
  \fontfamily{#3}\fontseries{#4}\fontshape{#5}%
  \selectfont}%
\fi\endgroup%
\begin{picture}(5724,2961)(64,-3013)
\put(2626,-2686){\makebox(0,0)[b]{\smash{{\SetFigFont{6}{7.2}{\rmdefault}{\mddefault}{\updefault}{\color[rgb]{0,0,0}$\frac{1}{2}$}%
}}}}
\put(1726,-211){\makebox(0,0)[b]{\smash{{\SetFigFont{6}{7.2}{\rmdefault}{\mddefault}{\updefault}{\color[rgb]{0,0,0}$0$}%
}}}}
\put(1651,-1786){\makebox(0,0)[b]{\smash{{\SetFigFont{6}{7.2}{\rmdefault}{\mddefault}{\updefault}{\color[rgb]{0,0,0}$2$}%
}}}}
\put(4126,-1486){\makebox(0,0)[b]{\smash{{\SetFigFont{6}{7.2}{\rmdefault}{\mddefault}{\updefault}{\color[rgb]{0,0,0}$4$}%
}}}}
\put(4426,-286){\makebox(0,0)[b]{\smash{{\SetFigFont{6}{7.2}{\rmdefault}{\mddefault}{\updefault}{\color[rgb]{0,0,0}$1$}%
}}}}
\put(4501,-1786){\makebox(0,0)[b]{\smash{{\SetFigFont{6}{7.2}{\rmdefault}{\mddefault}{\updefault}{\color[rgb]{0,0,0}$3$}%
}}}}
\put(751,-1261){\makebox(0,0)[b]{\smash{{\SetFigFont{6}{7.2}{\rmdefault}{\mddefault}{\updefault}{\color[rgb]{0,0,0}$s_1$}%
}}}}
\put(376,-961){\makebox(0,0)[b]{\smash{{\SetFigFont{6}{7.2}{\rmdefault}{\mddefault}{\updefault}{\color[rgb]{0,0,0}$4$}%
}}}}
\put(751,-1711){\makebox(0,0)[b]{\smash{{\SetFigFont{6}{7.2}{\rmdefault}{\mddefault}{\updefault}{\color[rgb]{0,0,0}${\geq}\frac{3}{4}$}%
}}}}
\put(1276,-811){\makebox(0,0)[b]{\smash{{\SetFigFont{6}{7.2}{\rmdefault}{\mddefault}{\updefault}{\color[rgb]{0,0,0}$\frac{1}{2}$}%
}}}}
\put(1276,-2161){\makebox(0,0)[b]{\smash{{\SetFigFont{6}{7.2}{\rmdefault}{\mddefault}{\updefault}{\color[rgb]{0,0,0}$\frac{1}{2}$}%
}}}}
\put(2101,-2011){\makebox(0,0)[b]{\smash{{\SetFigFont{6}{7.2}{\rmdefault}{\mddefault}{\updefault}{\color[rgb]{0,0,0}$s_2$}%
}}}}
\put(2101,-436){\makebox(0,0)[b]{\smash{{\SetFigFont{6}{7.2}{\rmdefault}{\mddefault}{\updefault}{\color[rgb]{0,0,0}$s_3$}%
}}}}
\put(3526,-1261){\makebox(0,0)[b]{\smash{{\SetFigFont{6}{7.2}{\rmdefault}{\mddefault}{\updefault}{\color[rgb]{0,0,0}$s_4$}%
}}}}
\put(4876,-2011){\makebox(0,0)[b]{\smash{{\SetFigFont{6}{7.2}{\rmdefault}{\mddefault}{\updefault}{\color[rgb]{0,0,0}$s_5$}%
}}}}
\put(4876,-436){\makebox(0,0)[b]{\smash{{\SetFigFont{6}{7.2}{\rmdefault}{\mddefault}{\updefault}{\color[rgb]{0,0,0}$s_6$}%
}}}}
\put(2926,-736){\makebox(0,0)[b]{\smash{{\SetFigFont{6}{7.2}{\rmdefault}{\mddefault}{\updefault}{\color[rgb]{0,0,0}$\frac{1}{2}$}%
}}}}
\put(4051,-736){\makebox(0,0)[b]{\smash{{\SetFigFont{6}{7.2}{\rmdefault}{\mddefault}{\updefault}{\color[rgb]{0,0,0}$\frac{1}{2}$}%
}}}}
\put(4126,-2236){\makebox(0,0)[b]{\smash{{\SetFigFont{6}{7.2}{\rmdefault}{\mddefault}{\updefault}{\color[rgb]{0,0,0}$\frac{1}{2}$}%
}}}}
\put(2851,-2161){\makebox(0,0)[b]{\smash{{\SetFigFont{6}{7.2}{\rmdefault}{\mddefault}{\updefault}{\color[rgb]{0,0,0}$\frac{1}{2}$}%
}}}}
\put(2476,-1186){\makebox(0,0)[b]{\smash{{\SetFigFont{6}{7.2}{\rmdefault}{\mddefault}{\updefault}{\color[rgb]{0,0,0}$\frac{1}{2}$}%
}}}}
\end{picture}%